\documentclass[runningheads]{llncs}
\usepackage[utf8]{inputenc}
\usepackage{comment}

\usepackage{tcolorbox}
\usepackage{amsmath}
\usepackage{enumerate}
\usepackage{graphicx}
\usepackage{amsmath}
\usepackage{amssymb}
\usepackage{comment}
\usepackage{hyperref}
\usepackage[english]{babel}

\usepackage{tikz}
\usetikzlibrary{shapes.geometric}
\usepackage{float}
\usepackage{tcolorbox}
\usepackage{graphicx}

\newtheorem{observation}{Observation}

\newtheorem{rr}{Reduction Rule}

\usepackage{import}

\def \calc {{\cal C}}
\def \calp {{\cal P}}
\def \calf {{\cal F}}

\def \calu {{\cal U}}
\def \calt {{\cal T}}

\def \hs {{{{\textsc{Hitting Set}}}}}

\def \tms {{{{\textsc{Terminal Monitoring Set}}}}} 
\def \tmss {{{{\textsc{TMS}}}}}

\def \atms {{{{\textsc{$\alpha$-Relaxed Terminal Monitoring Set}}}}} 
\def \atmss {{{{\textsc{$\alpha$-RTMS}}}}} 
\def \vc {{{{\textsc{Vertex Cover}}}}} 
\def \mcis {{\textsc{Multi Color Independent Set}}}

\def \hsg {{\textsc{Hitting Subgraphs in a Graph}}}
\def \hsgs {{\textsc{HSG}}}

\def \hpg {{\textsc{Hitting Paths in a Graph}}}
\def \hpgs {{\textsc{HPG}}}

\def \hpfb {{\textsc{Hitting Paths in a Flower with Budgets}}}
\def \hpfbs {{\textsc{HPFB}}}

\spnewtheorem{claimn}{Claim}{\bfseries}{\itshape}{\rmfamily}

\begin{document}

\title{The Parameterized Complexity of Terminal Monitoring Set}
\author{N. R. Aravind and Roopam Saxena}
\date{}

\institute{Department of Computer Science and Engineering\\ IIT Hyderabad, Hyderabad, India\\
\email{aravind@cse.iith.ac.in,cs18resch11004@iith.ac.in}
}

\maketitle

\begin{abstract}
In {\tms} ({\tmss}), the input is an undirected graph $G=(V,E)$, together with a collection $T$ of terminal pairs and the goal is to find a subset $S$ of minimum size that hits a shortest path between every pair of terminals.
We show that this problem is W[2]-hard with respect to solution size. On the positive side, we show that {\tmss} is fixed parameter tractable with respect to solution size plus distance to cluster, solution size plus neighborhood diversity, and feedback edge number. For the weighted version of the problem, we obtain a FPT algorithm with respect to vertex cover number, and for a relaxed version of the problem, we show that it is W[1]-hard with respect to solution size plus feedback vertex number.
\end{abstract}
\section{Introduction}

\subsection{Problem Definition}
Consider a communication network and a set of pairs of nodes, say 
$\calt=\{\{u_1,v_1\},\\\{u_2,v_2\},\dots \{u_m,v_m\}\}$ such that every pair in $\calt$ is communicating through a shortest path between them. We want to monitor this communicating data while deploying monitoring devices at minimum number of nodes in the network. Consider a scenario where we want to achieve an incremental deployment of a software defined network over a legacy network by deploying costly smart switches at only a few locations initially.
\paragraph{}

Motivated by the above scenarios, we formulate the terminal monitoring set problem as follows. 

\begin{tcolorbox}[colback=white]
{
{\tms} ({\tmss}): \newline
\textit{Input:} An instance $I$ = $(G,\calt,k)$, where $G=(V,E)$ is an undirected graph, $\calt=\{\{u_1,v_1\},\ldots,\{u_l,v_l\}\}$ where $u_i,v_i \in V$ and $k\in \mathbb{N}$.\newline
\textit{Output:} YES, if $G$ contains a set $S\subseteq V$ of size at most $k$ such that for every $\{u,v\}\in \calt$, there exists a $w \in S$ such that  $d(u,w)+d(w,v)=d(u,v)$ ; NO otherwise.
}
\end{tcolorbox}
\paragraph{}

We say a set $S\subseteq V(G)$ is a terminal monitoring set (regardless of its size) for $\calt$ if for every $\{u,v\}\in \calt$, there exists a $w \in S$ such that  $d(u,w)+d(w,v)=d(u,v)$.
We remark that a vertex can belong to multiple pairs in the list $\calt$, but we assume without loss of generality that no terminal pair appears twice. Further, we assume distance between every terminal pair in $\calt$ is finite (they are reachable by each other, hence belong to a same connected component).
\paragraph{}

While there are similar problems in the literature, some of which we will discuss in the next section, to the best of our knowledge, the optimization equivalent of {\tms} has not been studied before. 

\subsection{Related Work}

The closest related problem to {\tms} in the literature is the $(k,r)$-center problem, where we are given an undirected graph and it is asked if there exists a vertex set $S\subseteq V(G)$ such that $|S|\leq k$  and for every vertex $v\in V(G)\setminus S$ there exists a vertex $u\in S$ such that distance between $u$ and $v$ is at most $r$. 
Optimization of $r$ for a fixed $k$ is studied in \cite{ChechikP15,HochbaumS86,KhullerS00,Krumke95,DBLP:journals/jal/PanigrahyV98}, and optimization of $k$ for fixed $r$ is studied in \cite{Bar-IlanKP93,BrandstadtD98,CoelhoMW15,DBLP:conf/iwpec/LokshtanovMPRS13}. 
Recently Katsikarelis, Lampis, Paschos \cite{KatsikarelisLP19} studied parameterized complexity of $(k,r)$-center with respect to various structural parameters. Benedito, Melo, and Pedrosa \cite{BeneditoMP22} studied a problem related to $(k,r)$-center, under the name of Multiple Allocation $k$-Hub Center, which is also a closely related problem to {\tms}.
\paragraph{}

Another related problem is hub location where packets must travel from each source to its corresponding destination via a small number of hubs. Surveys of hub location can be found in \cite{AlumurK08} and \cite{FarahaniHAN13}.
\paragraph{}

The problem {\hs} is as follows, we are given a universe $\cal U$, and a family $\cal F$ of subsets of $\cal U$, and the question is to decide if there exists a set $S\subseteq \cal U$ of size at most $k$ such that $S$ intersects with every set in $\cal F$. Jansen in \cite{Jansen17} investigated a special case of {\hs} and called it {\hpg} where a graph $G$ and a collection $\calp$ of simple paths in $G$ is given, and it is to decide if there exist a vertex set of size at most $k$ in $G$ that intersects with every path in $\calp$. Jansen \cite{Jansen17} showed that {\hpg} is FPT parameterized by feedback edge number of the input graph, complementing this result Jansen also showed that {\hpg} is para NP-hard for the parameter feedback vertex number and treewidth of the input graph. 

\paragraph{}

{\tms} can be considered as a special case of {\hs}, where the universe $\cal U$ must be a vertex set of an input graph $G$, and every set in family $\cal F$ must be a union of vertices of all the shortest paths between a terminal pairs (vertex pairs). And hence, {\tms} can be formulated as a {\hs}, and thus it is natural that results on {\hs} are relevant to this work.
In particular, the results of \cite{Abu-Khzam10} and \cite{FlumG06} that use the sunflower lemma for set systems of bounded size. We use similar ideas in some of our algorithmic results.


\subsection{Our Results}

\begin{itemize}
    \item Hardness Results
    \begin{enumerate}[a]
        \item {\tms} is NP-hard.
        \item {\tms} is W[2]-hard with respect to solution size.
    \end{enumerate}
   
\end{itemize}

On the positive side, we have the following results.

\begin{itemize}
\item {\tms} admits an FPT algorithm when parameterized by solution size plus distance to cluster of the input graph $G$.
\end{itemize}

\begin{itemize}
\item {\tms} admits an FPT algorithm when parameterized by solution size plus neighborhood diversity of the input graph $G$.
\end{itemize}

\begin{itemize}
\item Weighted-{\tms} admits an FPT algorithm when parameterized by vertex cover number of the input graph $G$.
\end{itemize}

We also found that Jansen's Algorithm \cite{Jansen17} also solves {\tms} correctly, thus we have the following result.
\begin{itemize}
\item {\tms} admits an FPT algorithm when parameterized by the feedback edge number of the input graph $G$.
\end{itemize}

We leave open the parameterized complexity of {\tms} by feedback vertex number; however we obtain a hardness result for the following relaxation of {\tms}.

\begin{tcolorbox}[colback=white]
{
{\atms} ({\atmss}): \newline
\textit{Input:} An instance $I$ = $(G,{\cal T},k)$, where $G=(V,E)$ is an undirected graph, ${\cal T}=\{\{u_1,v_1\},\ldots,\{u_l,v_l\}\}$ where $u_i,v_i \in V$, $k\in \mathbb{N}$ and $\alpha \in {\mathbb {Q}}_{\geq 0}$.\newline
\textit{Output:} YES, if $G$ contains a subset $S$ of size at most $k$ such that for every $\{u,v\}\in \calt$, there exists a $w \in S$ such that $d(u,w)+d(w,v)\leq(1+\alpha)\cdot d(u,v)$ ; NO otherwise.
}
\end{tcolorbox}

\begin{itemize}
\item For every fixed $0< \alpha\leq 0.5$,  $\alpha$-RTMS is W[1]-hard with respect to feedback vertex number of the input graph $G$ plus solution size,  even when the input graph is a planar graph.
\end{itemize}

\section{Preliminaries}

\subsection{Graph Notations and Terminologies}
 All the graphs consider in this paper are simple and finite. We use standard graph notations and terminologies and refer the reader to \cite{DBLP:books/daglib/0030488} for basic graph notations and terminologies. We mention some notations used in this paper. A graph  $G=(V,E)$ has a vertex set $V$ and edge set $E$. We also use $V(G)$ and $E(G)$ to denote the vertex set and edge set of $G$ respectively. 
For an edge set $F\subseteq E(G)$, $V(F)$ denotes the set of all the vertices of $G$ with at least one edge in $F$ incident on it. For a vertex set $S\subseteq V(G)$, $G[S]$ denotes the induced sub graph of $G$ on vertex set $S$, and $G-S$ denotes the graph $G[V(G)\setminus S]$. For an edge set $E'\subseteq E$, $G[E']$ denotes the sub graph of $G$ on edge set $E'$ i.e. $G[E']=(V(E'),E')$. 
For an edge set $A\subseteq E(G)$, $G-A$ denotes the graph with vertex set $V(G)$ and edge set $E(G)\setminus A$.
 A component $C$ of a graph $G$ is a maximally connected subgraph of $G$.
 A graph $G$ is a forest if every component of $G$ is a tree.
 We use $G'\subseteq G$ to denote that $G'$ is a subgraph of $G$. For a $v\in V(G)$, we use $N_G(v)$ to denote the open neighborhood of $v$ in $G$ that is the set of all adjacent vertices of $v$ in $G$. Further, for a set $Z\subseteq V$, $N(Z)= \bigcup_{v\in Z} N(v)$. 
 \paragraph{}

For a weighted graph $G$, $w(e)$ is the weight of an edge $e\in E(G)$. For a graph $G$, $d(u,v)$ is the distance between vertices $u$ and $v$ in $G$. Given a graph $G$, and $u,v\in V(G)$, we define $SP_G(u,v)$ to be the set $\{x\mid x \in V(G)\ \land \ d(u,x)+d(x,v)=d(u,v)\}$. When the context is clear, we simply write $SP(u,v)$. For undirected graphs, $SP(u,v)$  and $SP(v,u)$ are the same; and we avoid writing $SP(\{u,v\})$.

\subsection{Set Terminologies}
We define the {\bf core} of a set family $\calf$ to be $\cap_{S \in \calf} S$ and denote it by $core(\calf)$. 
We say that a collection $\calf$ of sets forms a {\bf sunflower} if there is a set $C$ such that $S \cap T=C$ for every distinct pair $S,T\in \calf$. Notice that $C=core(\calf)$ in this case (see \cite{DBLP:books/sp/CyganFKLMPPS15,FlumG06} for details on sunflower, sunflower lemma, and its application to {\hs}). 

\subsection{Graph Structural Parameters}

A vertex set $S\subseteq V(G)$ is a vertex cover of $G$ if the sub graph $G[V\setminus S]$ has no edge. The minimum size of any vertex cover of $G$ is called vertex cover number of $G$. A vertex set $S\subseteq V(G)$ is a feedback vertex set of $G$ if the sub graph $G[V\setminus S]$ has no cycle. The minimum size of any feedback vertex set of $G$ is called feedback vertex number (FVN) of $G$. An edge set $F\subseteq E(G)$ is a feedback edge set of $G$ if the sub graph $G-F$ has no cycle. The minimum size of any feedback edge set of $G$ is called feedback edge number (FEN) of $G$.
\paragraph{}
A cluster graph is a disjoint union of complete graphs.
For a graph $G=(V,E)$, a vertex set $X\subseteq V$ is a cluster deletion set if $G-X$ is a cluster graph. The smallest size of a cluster deletion set of $G$ is called the distance to cluster of $G$.
A co-cluster graph is a complete multipartite graph. For a graph $G=(V,E)$, a vertex set $X\subseteq V$ is a co-cluster deletion set if $G-X$ is a co-cluster graph. The smallest size of a co-cluster deletion set of $G$ is called the
distance to co-cluster of $G$.
\paragraph{}
For the details on neighborhood diversity we refer to \cite{DBLP:journals/algorithmica/Lampis12}, and recall the definition here.

\begin{definition}[{\cite{DBLP:journals/algorithmica/Lampis12}}]
    In a graph $G=(V,E)$, two vertices $u$ and $v$ have the same type if and only if $N(v)\setminus \{u\}= N(u)\setminus \{v\}$.
    A graph $G=(V,E)$ has a neighborhood diversity at most $t$ if the set $V$ can be partitioned into at most $t$ sets such that all the vertices of each set have the same type.
\end{definition}

\subsection{Parameterized Complexity} 

\subsubsection{Basic Definitions and Terminologies}
For details on parameterized complexity, we refer to \cite{DBLP:books/sp/CyganFKLMPPS15,DBLP:series/txcs/DowneyF13}, and recall some definitions here.

\begin{definition}[\cite{DBLP:books/sp/CyganFKLMPPS15}]
A \textit{parameterized problem} is a language $L \subseteq \Sigma^* \times \mathbb{N} $ where $\Sigma$ is a fixed and finite alphabet. For an instance $I=(x,k) \in \Sigma^* \times \mathbb{N} $, $k$ is called the parameter. 
\end{definition}

\begin{definition}[\cite{DBLP:books/sp/CyganFKLMPPS15}]
    A parameterized problem is called \textit{fixed-parameter tractable} if there exists a computable function $f:\mathbb{N} \to \mathbb{N}$, a constant $c$, and an algorithm $\cal A$ (called a \textit{fixed-parameter algorithm} ) such that: the algorithm $\cal A$ correctly decides whether $I\in L$ in time bounded by $f(k).|I|^c$. The complexity class containing all fixed-parameter tractable problems is called FPT.
\end{definition}

Informally, a W[1]-hard problem is unlikely to be fixed parameter tractable, see \cite{DBLP:books/sp/CyganFKLMPPS15} for details on complexity class W[1].`
\begin{definition}[\cite{DBLP:books/sp/CyganFKLMPPS15}]
Let $P,Q$ be two parameterized problems. A parameterized reduction from $P$ to $Q$ is an algorithm which for an instance $(x,k)$ of $P$ outputs an instance $(x',k')$ of $Q$ such that:
\begin{itemize}
    \item $(x,k)$ is yes instance of $P$ if and only if $(x',k')$ is a yes instance of $Q$,
    \item $k'\leq g(k)$ for some computable function $g$, and    \item the reduction algorithm takes time $f(k)\cdot |x|^{O(1)}$ for some computable function $f$.
\end{itemize}
\end{definition}

\begin{theorem}[\cite{DBLP:books/sp/CyganFKLMPPS15}]\label{def-parameterized reduction}
If there is a parameterized reduction from $P$ to $Q$ and $Q$ is fixed parameter tractable then $P$ is also fixed parameter tractable.
\end{theorem}

\subsection{Some Problem Definitions and Existing Results}

\begin{tcolorbox}[colback=white]
{
{\hsg} ({\hsgs}): \newline
\textit{Input:} An instance $I$ = $(G,{\cal V},k)$, where $G=(V,E)$ is an undirected graph, $\cal V$ is a collection of subgraphs of $G$, and $k\in \mathbb{N}$.\\
\textit{Output:} YES, if $G$ contains a $S\subseteq V(G)$ of size at most $k$ that intersects with vertex set of every graph in $\cal V$; NO otherwise.
}
\end{tcolorbox}
\paragraph{}

If {\hsg} has the constraint that every graph in $\cal V$ is a simple path in $G$, then we call the problem {\sc{\hpg}} ({\hpgs}), the problem {\hpg} was studied in \cite{Jansen17}.

\begin{lemma}[folklore, discussed in \cite{Jansen17}]\label{lemma: folklore}
    {\hpg} can be solved in polynomial time if the input graph $G$ is a tree.
\end{lemma} 
\paragraph{}

From \cite{Jansen17} we recall that a graph $G$ is a flower graph if it has a specific vertex $z$ (called its core) such that $G-\{z\}$ is a disjoint union of paths, each such path is called its petal, and no internal vertex of any such path is adjacent to $z$, assume an arbitrary but distinct ordering $\{R_1,R_2,\dots R_l\}$ of these petals. 
\begin{tcolorbox}[colback=white]
{
{\hpfb}({\hpfbs})\cite{Jansen17}: \newline
\textit{Input:} An instance $I$ = $(G,z,P,b)$, where $G$ is a flower graph with core $z$ and petals $R_1,R_2,\dots R_l$, a set $P$ of simple paths in $G$, and $b:[l]\to \mathbb{N}_{\geq 1}$.\newline
\textit{Question:} Is there a set $S\subseteq V(G)\setminus \{z\}$ hitting every path in $P$ such that $|S\cap V(R_i)|=b(i)$ for every $i\in [l]$?.
}
\end{tcolorbox}

\begin{lemma}[\cite{Jansen17}]\label{lemma:flower}
    {\hpfbs} is polynomial time solvable.
\end{lemma}

\section{FPT Algorithms}

In this section we prove the following theorems.
\begin{theorem}\label{thm:clusterModulator}
There exists an FPT algorithm for {\tms} when parameterized by solution size plus distance to cluster of the input graph $G$.
\end{theorem}

\begin{theorem}\label{thm:ND}
There exists an FPT algorithm for {\tms} when parameterized by solution size plus neighborhood diversity of the input graph $G$.
\end{theorem}

\begin{theorem}\label{thm:vc}
There exists an FPT algorithm for {weighted-\tms}  when parameterized by vertex cover number of the input graph $G$.
\end{theorem}

\begin{theorem}\label{thm:FEN}
There exists an FPT algorithm for {\tms} when parameterized by the feedback edge number of the input graph $G$.
\end{theorem}
\paragraph{}

Given an instance $(G,\calt,k)$, our main idea is to create an equivalent instance $(\calu,\calf, k')$ of {\hs} where the number of sets in $\calf$ is a function of $k$ plus the structural parameter under consideration.
An instance $(\calu,\calf, k')$ of {\hs} with $\calf$ having $m$ sets and $\calu$ having $n$ elements can be solved in $O(2^m\cdot (n+m)^{O(1)})$ time by a standard dynamic programming algorithm (see \cite{DBLP:books/sp/CyganFKLMPPS15}). Using this we can obtain an FPT algorithm for {\tms}.
\paragraph{}

Another idea that we use is the following, which we will call the {\it standard vertex cover reduction}.
This is the reduction applied to obtain a quadratic kernel for {\vc} in \cite{BussG93}.
\begin{observation}\label{obs:vertexCover}
Given an instance $I=(\calu,\calf, k)$ of {\hs} where $\calf=\calf_1 \cup \calf_2$ such that every set in $\calf_2$ is of size at most $2$, then in time polynomial in $|I|$, either we conclude that $I$ is a no instance, or
we find a family $\calf_3$ of size at most $O(k^2)$ such that $(\calu,\calf,k)$ is a YES instance if and only if $(\calu, \calf_1 \cup \calf_3,k)$ is a YES instance.
\end{observation}
\paragraph{}

We remark that this can be generalized to subfamilies of bounded size and we indeed do this in a slightly more general way later in this paper (Proposition \ref{prop:hittingSet}).

\paragraph{}
Let $(G,\calt,k)$ be the input instance of {\tms}. We may assume that $G$ is connected otherwise we do the following. If a terminal pair exists in $\calt$ such that both of its vertices belong to distinct connected component, then we conclude that input is a no instance, as there is no path between them. Otherwise, we introduce a vertex $z$ and  arbitrarily pick a vertex from every connected component of $G$, we connect $z$ to every picked vertex. By this modification, no new shortest path is introduced between any terminal pair in $\calt$ as every connected component of original $G$ is now connected to rest of the graph by only a single edge, hence no new cycle in $G$ is introduced. After this modification,  distance to cluster of $G$ is increased by at most $1$, vertex cover number  is increased by at most $1$, and feedback edge number remains same. Also observe that if neighborhood diversity of $G$ was $t$, then after this modification the neighborhood diversity of $G$ will be at most $2t+1$, which is sufficient for our purposes. Hence, we may assume that the input graph $G$ is always connected.

\subsection{Parameterized by Solution Size + Distance to Cluster}\label{sec:clustermodulator}
In this section, we give the proof of Theorem \ref{thm:clusterModulator}.
\paragraph{}

Given an input instance $(G,\calt,k)$ of {\tms}, 
We say that a family $H\subseteq \calt$ is \textit{core-invariant} if  family $\{SP(u,v)\mid \{u,v\} \in H\}$ forms a sunflower.

\begin{rr} \label{rr:1}
    Given a core-invariant family $H$ with at least $k+2$ pairs, remove all but $k+1$ pairs of $H$ from $\calt$.
\end{rr}
\paragraph{}

 The above rule is safe because any set of size at most $k$ that hits $(k+1)$ sets in the sunflower must hit the core, and hence hit $SP(u,v)$ for every pair $\{u,v\}$ in $H$.

\paragraph{}
For the input graph $G=(V,E)$, let $M\subseteq V$ be a cluster deletion set of size at most $q$. Let these clique components be $C_1,C_2\ldots,C_r$ and let $\calc=\bigcup_{i\in[r]}V(C_i)$.
Further, we define the following.
\begin{itemize}
    \item $T_0=\{\{u,v\} \mid \{u,v\}\in (\calt \cap E)\}$;
    \item $T_1=\{\{u,v\} \mid \{u,v\}\in (\calt \setminus T_0)\ \land \ u\in \calc\ \land \ v \in M\}$; 
    \item $T_2=\{\{u,v\} \mid \{u,v\}\in (\calt \setminus T_0) \ \land \ u,v \in \calc\}$.
    \item $T_3=\{\{u,v\} \mid \{u,v\}\in (\calt \setminus T_0) \ \land \ u,v \in M\}$.
\end{itemize}

Observe that $\{T_0,T_1,T_2,T_3\}$ is a partition of $\calt$.
We first reduce and bound the size of $T_2$.

\subsubsection{Reducing and Bounding Size of $T_2$.}
We fix an arbitrary ordering between both the vertices of every pair $\{u,v\}\in T_2$, and thus $\{u,v\}$ is denoted by $(u,v)$ or $(v,u)$ depending on the ordering. The ordering will not affect $SP(u,v)$ for $\{u,v\}$. 
\paragraph{}
We say that a pair $(u,v)\in T_2$ is of \textit{type} $X$ if $X$ is the smallest subset of $M \times \{1,2\} \times M \times \{1,2\}$ such that the following holds: for every shortest path $P$ from $u$ to $v$, there exists $(x,i,y,j) \in X$ such that $x$ and $y$ are the closest vertices in $V(P) \cap M$ to $u$ and $v$ respectively, and $d(u,x)=i$, $d(v,y)=j$.
\paragraph{}
Note that the number of possible types $X$ is at most $2^{4q^2}$. Every pair $(u,v)\in T_2$ is of some type $X$. This is because $u$ and $v$ belong to distinct cliques, and thus every shortest path from $u$ to $v$ must go through the vertices of $M$, and for every vertex in $C$ its distance from its closest vertex in $M$ is at most $2$. Further, given $(u,v)$ and $X$, we can verify in polynomial time whether $(u,v)$ is of type $X$ by verifying the condition that a tuple $(x,i,y,j) \in M \times \{1,2\} \times M \times \{1,2\}$ belongs to $X$ if and only if $d(u,x)=i$, $d(y,v)=j$, and $d(x,y)+i+j= d(u,v)$.
\paragraph{}

For each $X\subseteq M \times \{1,2\} \times M \times \{1,2\}$, we define an auxiliary graph $H_X=(V_X,E_X)$ with $V_X=\{c_1,c_2,\ldots,c_r\}$ and $\{c_i,c_j\} \in E_X$ if there exists a pair $(u,v)\in T_2$ of type $X$ such that one of its endpoint (either $u$ or $v$) belong to $C_i$ and the other endpoint belong to $C_j$.
These auxiliary graphs $H_X$ can be computed in polynomial time.

\begin{claimn}\label{claim:spsubset}
    For every $(u,v)\in T_2$ and its type $X$, the following holds: $SP(u,v)\supseteq \bigcup_{(x,i,y,j)\in X} SP(x,y)$.
\end{claimn}
\begin{proof}
    By definition, for every $(x,i,y,j)\in X$, there exists a shortest path $P$ from $u$ to $v$ such that $x,y\in V(P)$, and thus every vertex in $SP(x,y)$ must belong to $SP(u,v)$. 
\end{proof}

\begin{claimn}\label{clm:edges}
    If $|E_X|>(2(k+2))^3$, then we can find a core-invariant family with at least $(k+2)$ pairs, and hence apply Reduction Rule \ref{rr:1}.
\end{claimn}

\begin{proof}
    If $|E_X|>(2(k+2))^3$, then $H_X$ must contain a vertex of degree at least $(2(k+2))^2$ or (by Vizing's theorem), a matching $A$ of size at least $2(k+2)$.\\
    {\bf{Case 1}}: $H_X$ contains a matching $A$ of size at least $2(k+2)$. In this case, for every edge in $A$, pick exactly one terminal pair from $T_2$ which corresponds to its construction, let $A^*$ be these picked pairs, every pair in $A^*$ is of type $X$. We claim that $A^*$ is a core-invariant family with at least $k+2$ pairs, to see this consider the following arguments.
    By Claim \ref{claim:spsubset}, we have that for every distinct $(u_1,v_1),(u_2,v_2)\in A^*$, $SP(u_1,v_1)\cap SP(u_2,v_2)\supseteq \bigcup_{(x,i,y,j)\in X} SP(x,y)$; the equality holds since every vertex in $\bigcup_{\{u,v\}\in A^*}\{u,v\}$ belongs to a distinct cluster and if there exist a vertex $w$ in  $SP(u_1,v_1)\cap SP(u_2,v_2)$ which do not belong to $\bigcup_{(x,i,y,j)\in X} SP(x,y)$, then $w$ must belong to the clique of one of $u_1,v_1,u_2,v_2$. W.l.o.g. let $w$ belong to the clique of $u_1$, in this case $w$ can not belong to $SP(u_2,v_2)$ as it do not belong to $\bigcup_{(x,i,y,j)\in X} SP(x,y)$, a contradiction. Hence, for every distinct $(u_1,v_1),(u_2,v_2)\in A^*$, $SP(u_1,v_1)\cap SP(u_2,v_2)= \bigcup_{(x,i,y,j)\in X} SP(x,y)$.
    Hence, using $A^*$ we can apply Reduction Rule \ref{rr:1}.
    
    {\bf{Case 2}}: $H_X$ contains a vertex $c_l$ of degree at least $(2(k+2))^2$. For every edge incident on $c_l$, pick exactly one terminal pair from $T_2$ which corresponds to its construction, let $B$ be the set of these picked pairs. Let $V_B= \bigcup_{\{u,v\}\in B}\{u,v\}$.
    Consider an auxiliary graph $G_B=(V_B,B)$, here the ordering of pairs of $B$ is not considered, thus $G_B$ is un-directed. By construction, one endpoint of every edge in $G_B$ belongs to clique $C_l$ and other endpoint belong to a distinct maximal clique of $G-M$.
    Since $G_B$ has at least $(2(k+2))^2$ edges, there must either be a vertex $u$ in $G_B$ which belong to clique $C_l$ and has degree $2(k+2)$ in $G_B$, or there must be a matching $M_B$ in $G_B$ of size $2(k+2)$. We discuss both the cases separately in the following.
    
    {\bf{Case 2a}}: A vertex $u$ in $G_B$ belonging to $C_l$ has degree $2(k+2)$ in $G_B$. Let $U$ be the set of edges (terminal pairs in $T_2$) incident on $u$ in $G_B$. We recall that the pairs in $T_2$ were assigned arbitrary ordering; let $U_1$ be those terminal pairs of $U$ where the first vertex is $u$ (Ex. $(u,v)$), and let $U_2$ be those terminal pairs in $U$ where the second vertex is $u$ (Ex. $(x,u)$). 
    At least one of $U_1$ or $U_2$ has at least $k+2$ pairs, let it be $U_1$. For every pair in $U_1$ the second vertex belongs to a distinct maximal clique in $G-M$. We claim for every distinct $(u,v_1),(u,v_2)\in U_1$, $SP(u,v_1)\cap SP(u,v_2)= \bigcup_{(x,i,y,j)\in X} (SP(u,x)\cup SP(x,y))$ and thus $U_1$ is core-invariant family. It is straight forward to verify that for every distinct $(u,v_1),(u,v_2)\in U_1$, $SP(u,v_1)\cap SP(u,v_2)\supseteq \bigcup_{(x,i,y,j)\in X} (SP(u,x)\cup SP(x,y))$, we move on to show equality. Let a vertex $w$ in $SP(u,v_1)\cap SP(u,v_2)$ exists which do not belong to $\bigcup_{(x,i,y,j)\in X} (SP(u,x)\cup SP(x,y))$, then $w$ must belong to the maximal clique of one of $v_1,v_2$. W.l.o.g. let $w$ belong to the maximal clique of $v_1$, in this case $w$ can not belong to $SP(u,v_2)$ as it do not belong to $\bigcup_{(x,i,y,j)\in X} (SP(u,x)\cup SP(x,y))$, a contradiction. Similar arguments holds when $U_2$ has at least $k+2$ pairs. Hence, using $U_1$ or $U_2$ we apply Reduction Rule \ref{rr:1}.
    
    {\bf{Case 2b}}: There exists a matching $M_B$ in $G_B$ of size $2(k+2)$, then for every edge in $M_B$, its one endpoint is a distinct vertex of clique $C_l$ and the other endpoint is a vertex of a distinct maximum clique in $G-M$. Let $M_{B,1}$ (resp $M_{B,2}$) be the set of those terminal pairs of $M_B$ where the first vertex (resp. second vertex) belongs to $C_l$. Either $M_{B,1}$ or $M_{B,2}$ has at least $k+2$ pairs, let it be $M_{B,1}$. 
     Let $V_l=\bigcup_{(u,v)\in M_{B,1}}\{u\}$ that is $V_l$ consists of all the vertices of $C_l$ which belong to a pair in $M_{B,1}$. For every $(u,v)\in M_{B,1}$, we define $L_u^*=(\bigcup_{(x,i,y,j)\in X} SP(u,x) \cap V(C_l))\setminus V_l$. we claim that for every distinct $(u_1,v_1),(u_2,v_2)\in M_{B,1}$, $L_{u_1}^* = L_{u_2}^*$, this is because every $L_u^*$ contains only those vertices of $C_l$ which do not belong to any pair in $M_{B,1}$ and are adjacent to an $x$ such that $(x,i,y,j)\in X$ because the distance between $u$ and $x$ is at most $2$. Thus, we now use $L^*$ instead of $L_{u}^*$ as the value $L_{u}^*$ for every $(u,v)\in M_{B,1}$ is same. We now claim that for every distinct $(u_1,v_1),(u_2,v_2)\in M_{B,1}$, $SP(u_1,v_1)\cap SP(u_2,v_2)= L^*\cup (\bigcup_{(x,i,y,j)\in X} SP(x,y))$ and thus $M_{B,1}$ is a core-invariant family. For the proof, let $w$ be a vertex in $SP(u_1,v_1)\cap SP(u_2,v_2)$ which do not belong to $L^*\cup (\bigcup_{(x,i,y,j)\in X} SP(x,y))$, then $w$  must belong to either $C_l$ or to the maximal clique of either $v_1$ or $v_2$. If $w$ belong to the maximal clique of either $v_1$ then it can not belong to $SP(u_2,v_2)$, similarly if $w$ belong to the maximal clique of either $v_2$ then it can not belong to $SP(u_1,v_1)$, a contradiction in both cases. We now consider the case when $w$ belongs to $C_l$, in this case $w$ belongs to both $(\bigcup_{(x,i,y,j)\in X} SP(u_1,x))$ and $(\bigcup_{(x,i,y,j)\in X} SP(u_2,x))$, and since $w$ is not in $L^*$, $w$ must belong to $V_l$.
     Further, since $u_1\neq u_2$ we have that $w\neq u_1$ and $w\neq u_2$. In this case there exists a pair $(w,w')\in M_{B,1}$, and thus $d(w,x)=d(u_1,x)=d(u_2,x)=i$ for every $(x,i,y,j)\in X$, contradicting that $w$ belongs to $\bigcup_{(x,i,y,j)\in X}SP(u_1,x)$. It follows that $M_{B,1}$ is core-invariant with at least $k+2$ pairs, similar arguments holds when $M_{B,2}$ has at least $k+2$ pairs. Hence, using at least one of  $M_{B,1}$ and $M_{B,2}$ we can apply Reduction Rule \ref{rr:1}.
\end{proof}

\begin{claimn}\label{clm:vertices}
    If there exist distinct $p,q\in[r]$ such that there are more than $(2(k+2))^2$ distinct pairs in $T_2$ of a same type $X$ with one endpoint in $C_p$ and another endpoint in $C_q$, then we can find a core-invariant family with at least $k+2$ pairs. Hence, apply Reduction Rule \ref{rr:1}.
\end{claimn}

\begin{proof}
    Let $B$ be the set of all the pairs in $T_2$ of type $X$ with one endpoint in $C_p$ and another endpoint in $C_q$ for distinct $p,q\in[r]$.
    Let $V_B= \bigcup_{\{u,v\}\in B}\{u,v\}$.
    Consider an auxiliary graph $G_B=(V_B,B)$, here the ordering of pairs of $B$ is not considered, thus $G_B$ is undirected. By construction, one endpoint of every edge in $G_B$ belongs to clique $C_p$ and other endpoint belong to $C_q$ hence $G_B$ is a bipartite graph.
    Since $G_B$ has at least $(2(k+2))^2$ edges, there must either be a vertex $u$ in $G_B$ which has degree $2(k+2)$ in $G_B$, or there must be a matching $M_B$ in $G_B$ of size $2(k+2)$. We discuss both the cases separately in the following.

    {\bf{Case 1}}: A vertex $u$ in $G_B$ has degree $2(k+2)$ in $G_B$. Let $u$ belong to $C_p$ and the case when $u$ belong to $C_q$ follows similar arguments. Let $U$ be the set of edges (terminal pairs in $T_2$) incident on $u$ in $G_B$. We recall that the pairs in $T_2$ were assigned arbitrary ordering; let $U_1$ be those terminal pairs of $U$ where the first vertex is $u$ (Ex. $(u,v)$), and let $U_2$ be those terminal pairs in $U$ where the second vertex is $u$ (Ex. $(x,u)$). 
    At least one of $U_1$ or $U_2$ has at least $k+2$ pairs, let it be $U_1$. For every pair in $U_1$ the second vertex belongs to $C_q$. Let $V_q=\bigcup_{(u,v)\in U_1}\{v\}$ that is $V_q$ consists of all the vertices of $C_q$ which belong to a pair in $U_1$. For every $(u,v)\in U_1$, we define $R_v^*=(\bigcup_{(x,i,y,j)\in X} SP(y,v) \cap V(C_q))\setminus V_q$. we claim that for every distinct $(u,v_1),(u,v_2)\in U_1$, $R_{v_1}^* = R_{v_2}^*$, this is because every $R_u^*$ contains only those vertices of $C_q$ which do not belong to any pair in $U_1$ and are adjacent to a $y$ such that $(x,i,y,j)\in X$, and because the distance between $v$ and $y$ is at most $2$. Thus, we now use $R^*$ instead of $R_{v}^*$ as the value $R_{v}^*$ for every $(u,v)\in U_1$ is same. We now claim that for every distinct $(u,v_1),(u,v_2)\in U_1$, $SP(u,v_1)\cap SP(u,v_2)=  (\bigcup_{(x,i,y,j)\in X} (SP(u,x)\cup SP(x,y)))\cup R^*$, and thus $U_1$ is a core-invariant family. For the proof, let $w$ be a vertex in $SP(u,v_1)\cap SP(u,v_2)$ which do not belong to $(\bigcup_{(x,i,y,j)\in X} (SP(u,x)\cup SP(x,y)))\cup R^*$, then $w$  must belong to $C_q$, in this case $w$ belongs to both $(\bigcup_{(x,i,y,j)\in X} SP(y,v_1))$ and $(\bigcup_{(x,i,y,j)\in X} SP(y,v_2))$, and since $w$ is not in $R^*$, $w$ must belong to $V_q$.
     Further, since $v_1\neq v_2$ we have that $w\neq v_1$ and $w\neq v_2$. In this case there exists a pair $(w',w)\in U_1$, and thus $d(y,w)=d(y,v_1)=d(y,v_2)=j$ for every $(x,i,y,j)\in X$, contradicting assumption that $w$ belongs to $\bigcup_{(x,i,y,j)\in X}SP(y,v_1)$. It follows that $U_1$ is core-invariant with at least $k+2$ pairs, similar arguments holds when $U_2$ has at least $k+2$ pairs. Hence, using at least one of  $U_1$ and $U_2$ we can apply Reduction Rule \ref{rr:1}.

    {\bf{Case 2:}} There exists a matching $M_B$ in $G_B$ of size $2(k+2)$. Then for every edge in $M_B$, its one endpoint is a distinct vertex of clique $C_p$ and the other endpoint is a distinct vertex $C_q$. Let $M_{B,1}$ (resp $M_{B,2}$) be the set of those terminal pairs of $M_B$ where the first vertex (resp. second vertex) belongs to $C_l$. Either $M_{B,1}$ or $M_{B,2}$ has at least $k+2$ pairs, let it be $M_{B,1}$. Let $V_p=\bigcup_{(u,v)\in M_{B,1}}\{u\}$, and  $V_q=\bigcup_{(u,v)\in M_{B,1}}\{v\}$. And for every $(u,v)\in M_{B,1}$, we define $L_u^*=(\bigcup_{(x,i,y,j)\in X} SP(u,x) \cap V(C_p))\setminus V_p$, and $R_v^*=(\bigcup_{(x,i,y,j)\in X} SP(y,v) \cap V(C_q))\setminus V_q$. We claim that for every distinct $(u_1,v_1),(u_2,v_2)\in M_{B,1}$, $L_{u_1}^* = L_{u_2}^*$ and $R_{v_1}^* = R_{v_2}^*$, this is because every $L_u^*$ (resp. $R_v^*$) contains only those vertices of $C_p$ (resp. $C_q$ which do not belong to any pair in $M_{B,1}$ and are adjacent to an $x$ (resp. $y$) such that $(x,i,y,j)\in X$, and because the distance between $u$ and $x$, and $y$ and $v$ is at most $2$. Thus, we now use $L^*$ (resp $R^*$) instead of $L_{u}^*$  (resp. $R_v^*$) as the value $L_{u}^*$ (resp. $R_{v}^*$) for every $(u,v)\in M_{B,1}$ is same. 
\paragraph{}
    We now claim that for every distinct $(u_1,v_1),(u_2,v_2)\in M_{B,1}$, $SP(u_1,v_1)\cap SP(u_2,v_2)= L^*\cup (\bigcup_{(x,i,y,j)\in X} SP(x,y)) \cup R^*$ and thus $M_{B,1}$ is a core-invariant family. For the proof, let $w$ be a vertex in $SP(u_1,v_1)\cap SP(u_2,v_2)$ which do not belong to $L^*\cup (\bigcup_{(x,i,y,j)\in X} SP(x,y))\cup R^*$, then $w$  must belong to either $C_p$ or to $C_q$. If $w$ belongs to $C_p$, then $w$ belongs to both $(\bigcup_{(x,i,y,j)\in X} SP(u_1,x))$ and $(\bigcup_{(x,i,y,j)\in X} SP(u_2,x))$, and since $w$ is not in $L^*$, $w$ must belong to $V_p$.
     Further, since $u_1\neq u_2$ we have that $w\neq u_1$ and $w\neq u_2$. In this case there exists a pair $(w,w')\in M_{B,1}$, and thus $d(w,x)=d(u_1,x)=d(u_2,x)=i$ for every $(x,i,y,j)\in X$, contradicting that $w$ belongs to $\bigcup_{(x,i,y,j)\in X}SP(u_1,x)$. Else if $w$ belongs to $C_q$, then $w$ belongs to both $(\bigcup_{(x,i,y,j)\in X} SP(y,v_1))$ and $(\bigcup_{(x,i,y,j)\in X} SP(y,v_2))$, and since $w$ is not in $R^*$, $w$ must belong to $V_q$.  Further, since $v_1\neq v_2$ we have that $w\neq v_1$ and $w\neq v_2$. In this case there exists a pair $(w',w)\in M_{B,1}$, and thus $d(y,w)=d(y,v_1)=d(y,v_2)=i$ for every $(x,i,y,j)\in X$, contradicting that $w$ belongs to $\bigcup_{(x,i,y,j)\in X}SP(u_1,x)$.
     It follows that $M_{B,1}$ is core-invariant with at least $k+2$ pairs, similar arguments holds when $M_{B,2}$ has at least $k+2$ pairs. Hence, using at least one of  $M_{B,1}$ and $M_{B,2}$ we can apply Reduction Rule \ref{rr:1}.
\end{proof}
\paragraph{}

For every type $X$, we apply Reduction Rule \ref{rr:1} exhaustively while the condition in Claim \ref{clm:edges} or the condition in Claim \ref{clm:vertices} holds.
We now have $|E_X| \leq (2(k+2))^3$. Further, for every edge $\{i,j\} \in E_X$, the number of pairs in $T_2$ with one endpoint in $C_p$ and the other end-point in $C_q$ is at most $(2(k+2))^2$ (because Claim 2 is not applicable).
Thus, we obtain a reduced set $T_2$ such that 
$|T_2| \leq (2(k+2))^5 2^{4q^2}$.

\subsubsection{Reducing and Bounding Size of $T_1$.}
We shall now reduce the size of $T_1$ in a similar manner. We fix an ordering between both the vertices of every pair $\{u,v\}\in T_1$ such that its first vertex belong to $\calc$, and thus $\{u,v\}$ denoted by $(u,v)$ when $u$ belongs to $\calc$ and $v$ belongs to $M$. We say that a pair $(u,v)\in T_1$ is of \textit{type} $X$ if $X$ is the smallest subset of $M \times \{1,2\}$ such that the following holds: for every shortest path $P$ from $u$ to $v$, there exists $(x,i) \in X$ such that $x$ is the closest vertex in $V(P) \cap M$ to $u$ and $d(u,x)=i$.
\paragraph{}

As before, for every $X\subseteq M\times \{1,2\}$, we define an auxiliary bipartite graph $H_X=(V_X,E_X)$, where $V_X= \{c_1,c_2,\ldots,c_r\} \cup M$ and $\{c_i,m\} \in E_X$ if and only if there exists a pair $(u,m)\in T_1$ of type $X$ such that $u \in C_i$.

\begin{claimn}\label{clm:T1Degree}
If there is a vertex $m\in M$ with at least $(k+2)$ neighbors in $H_X$, then we can find a core-invariant family with at least $(k+2)$ pairs and apply Reduction Rule \ref{rr:1}.
\end{claimn}

\begin{claimn}\label{clm:T1vertices}
     If there exists $i\in[r]$, $m\in M$, and $X\subseteq M\times\{1,2\}$ such that there are at least $(k+2)$ distinct pairs in $T_1$ of type $X$ with one endpoint in $C_i$ and other endpoint being vertex $m$, then we can find a core-invariant family at least $k+2$ pairs, and apply Reduction Rule \ref{rr:1}.
\end{claimn}

\paragraph{}

The above claims are similar to that of Claim \ref{clm:edges} and Claim \ref{clm:vertices}, hence we skip their proof. For every type $X\subseteq M\times\{1,2\}$, we apply Reduction Rule \ref{rr:1} exhaustively while the condition in Claim \ref{clm:T1Degree} or the condition in Claim \ref{clm:T1vertices} holds. We now have: $|T_1| \leq q(k+2)^2 2^{2q}$.

\subsubsection{Reducing and Bounding Size of $T_0$ and $T_3$.}
\paragraph{}

For every pair in $T_3$, both the vertices belong to $M$ and since $|M|=q$, there can be at most $O(q^2)$ pairs in $T_3$.
\paragraph{}

For $i\in \{0,1,2,3\}$, we define $\calt_i=\{SP(u,v)\mid\{u,v\} \in T_i\}$. We construct an instance $I=(V(G),\calt_1\cup\calt_2\cup\calt_3\cup\calt_0,k)$ of {\hs} which is an equivalent instance of $(G,\calt,k)$ of {\tms}. Observe that for every pair $\{u,v\}\in T_0$, $SP(u,v)=\{u,v\}$ as $u,v$ are neighbors, thus every set in $\calt_0$ is of size at most $2$. Thus, we apply Observation \ref{obs:vertexCover}, and either conclude that $I$ is a no instance or obtain an equivalent instance $I'=(V(G),\calt_1\cup\calt_2\cup\calt_3\cup\calt_0^*,k)$ where $|\calt_0^*|=O(k^2)$.

\paragraph{}

Thus, the number of sets in $\calt_1\cup\calt_2\cup\calt_3\cup\calt_0^*$ are bounded by $O(q(k+2)^5 2^{4q^2})$.
We then solve the {\hs} instance $I'$ in time FPT in $q+k$. This completes the proof of Theorem \ref{thm:clusterModulator}.

\subsection{Parameterized by Neighborhood Diversity}
In this section we give proof of Theorem \ref{thm:ND}. 
For the proof we use existing kernels for the {\hs} problem with a small modification.

\begin{definition}
For a family $\calf$ and a set $S\in \calf$, we define the \textit{effective size} of $S$ with respect to $\calf$ as $\Upsilon(\calf,S)=|S \setminus core(\calf)|$. Further, we define the \textit{effective size bound} of $\calf$ as $\Upsilon(\calf)= \textsc{max}\{\Upsilon(\calf,S)\mid S\in \calf\}$.
\end{definition}
\paragraph{}

The following proposition is an adaptation of a well-known application of the sunflower lemma to {\hs}, see \cite{Abu-Khzam10,DBLP:books/sp/CyganFKLMPPS15,FlumG06}.

\begin{proposition}\label{prop:hittingSet}
There is an algorithm, that given set families $\calf_1,\ldots,\calf_m$ each with an effective size bound at most $d$, and an integer $k$, finds set families $\calc_1,\ldots,\calc_m$ in time polynomial in $m+d$, and the number of sets such that:
\begin{itemize}
    \item $|\calc_i|=O(k^d\cdot d!)$;
    \item $(\calu,\bigcup_{i\in[m]}\calf_i,k)$ and $(\calu,\bigcup_{i\in[m]}\calc_i,k)$ are equivalent instances of {\hs}.
\end{itemize}
\end{proposition}
\paragraph{}

Let $(G,\calt,k)$ be an input to {\tms} with $G$ having neighborhood diversity $t$. We assume that $G$ is connected. Let $V(G)=V_1 \cup \ldots V_t$ where all the vertices of  each $V_i$ have the same type. Observe that each $V_i$ either forms  a clique or an independent set in $G$.
\paragraph{}

For $1\leq i \leq j \leq t$, let $T_{i,j}=\{\{u,v\} \mid \{u,v\}\in \calt \ \land \ u\in V_i \ \land \ v \in V_j\}$. We define the family $\calf_{i,j}=\{SP(u,v)\mid \{u,v\} \in T_{i,j}\}$, and construct $I= (V(G),\bigcup_{1\leq i\leq j\leq t}\calf_{i,j},k)$
Observe that every set in every $\calf_{i,j}$ contains one vertex from $V_i$ and one vertex from $V_j$; and thus, has an effective size bound $\leq 2$, where $core(\calf_{i,j})\supseteq\left(\cup_{S \in \calf_{i,j}} S\right) \setminus (V_i \cup V_j)$. Thus, using Proposition \ref{prop:hittingSet}, we obtain an equivalent instance of {\hs} with $O(t^2k^2)$ sets, and solve it in time FPT in $t+k$. This finishes the proof of Theorem \ref{thm:ND}.

\subsubsection{Alternative proof of Theorem \ref{thm:ND}}
\paragraph{}
For the proof, we follow terminologies used in Section \ref{sec:clustermodulator}.
For $1\leq i \leq j \leq t$, let $T_{i,j}=\{\{u,v\} \mid \{u,v\}\in \calt \ \land \ u\in V_i \ \land \ v \in V_j\}$, further we define graph $H_{i,j}= (V_i\cup V_j, T_{i,j})$. 

\begin{claimn}\label{clm:ND2}
    For $1\leq i \leq j\leq t$,  If $|T_{i,j}|>(k+2)^2$, then we can find a core-invariant family with at least $(k+2)$ pairs, and hence apply Reduction Rule \ref{rr:1}.
\end{claimn}
\begin{proof}
    By Vizing's theorem, $H_{i,j}$ will either have a vertex of degree at least $k+2$ or a matching of size $k+2$.\\
    \textbf{Case 1:} There is a vertex $u$ with degree at least $k+2$. In this case, consider the following two cases. \\
    \textbf{Case 1a:} if $i\neq j$. W.l.o.g. let $u$ belong to $V_i$, in this case all the neighbors of $u$ in $H_{i,j}$ belong to $V_j$ and they have the same type in $G$, the pairs containing $u$ in $T_{i,j}$ forms a core invariant family.\\
    \textbf{Case 1b:} If $i=j$. In this case if $V_i$ forms a clique, then all the neighbors of $u$ in $H_{i,i}$ are also neighbors of $u$ in $G$, hence the pairs containing $u$ in $T_{i,i}$ forms a core invariant family. If $V_i$ forms an independent set, then $u$ and all the neighbors of $u$ in $H_{i,i}$ forms an independent set and have the same type in $G$, hence the pairs containing $u$ in $T_{i,i}$ forms a core invariant family.\\
     \textbf{Case 2:} There is a matching $A$ of size least $k+2$. In this case,  Consider the following two cases.\\
     \textbf{Case 2a:} If $i\neq j$, then for every pair in $A$, one endpoint is a distinct vertex of $V_i$ and other endpoint is a distinct vertex of $V_j$, since all the vertices of $V_i$ (resp. $V_j$) have the same type, $A$ forms a core invariant family.\\
     \textbf{Case 2b:} If $i=j$. In this case all the pairs in $A$ are mutually disjoint, and all the vertices in $A$ have the same type, $A$ forms a core invariant family.
\end{proof}
\paragraph{}

If condition in Claim \ref{clm:ND2} is no longer applicable, then for every $1\leq i\leq j \leq t$, $|T_{i,j}|\leq O(k^2)$, and hence $|\bigcup_{1\leq i\leq j\leq t} T_{i,j}| =O(t^2k^2)$. Hence, we obtain an equivalent instance of {\hs} with $O(t^2k^2)$ sets, and solve it in time FPT in $t+k$.

\subsection{Weighted Version Parameterized by Vertex Cover Number }
In this section we give proof of Theorem \ref{thm:vc}.
In this section, we consider weighted-{\tms} where the underlying graph $G$ has positive weights on its edges. The distance is then the shortest weighted distance and $SP(u,v)$ is defined accordingly.
\paragraph{}

Let $(G,\calt,k)$ be the input instance of weighted-{\tms}, and let $C=\{v_1,\ldots,v_t\}$ be a given vertex cover of $G$ and $I=V \setminus C$.
Note that we may assume that $k \leq t$; otherwise $C$ itself is a terminal monitoring set of size at most $k$.
\paragraph{}

Let $M$ be a binary matrix indexed by $C$, a terminal monitoring set $S$ of $\calt$ is called $M$-compatible if the following holds:
\begin{equation}\label{eq:weightedvc}
\forall u,v \in C, S \cap SP(u,v) \neq \emptyset \Leftrightarrow M[u,v]=1. 
\end{equation}
\paragraph{}

We note that by the above condition, a vertex $v \in S$ if and only if $M[v,v]=1$.
Observe that for every terminal monitoring set $S$, there exists a binary $M$ matrix indexed by $C$ by such that $S$ is $M$-compatible. The idea is to first guess a matrix $M$, and then verify in desired time if an $M$-compatible terminal set $S$ of size at most $k$ exists or not; we make exhaustive guesses on $M$, and if for any $M$, an $M$-compatible terminal set of size at most $k$ exist, then we conclude that $(G,\calt,k)$ is a yes instance, otherwise we conclude that its a no instance.
\paragraph{}

Given an $M$, to verify if an $M$-compatible terminal set $S$ of size at most $k$ exist, we construct and solve an instance of {\hs} in time FPT in $|C|$.
\paragraph{}

To this end, first we guess $M$.
and let $C_1=\{v\mid v\in C \ \land \  M[v,v]=1\}$ and $C_0=C \setminus C_1$.
Let $S_0=\bigcup_{u,v\in C \land M[u,v]=0}SP(u,v)$.
Note that for the guessed $M$, we are looking for a terminal monitoring set $S$ such that $S \cap S_0=\emptyset$. If there exist $u,v \in C$ such that $M[u,v]=1$ and $SP(u,v)\setminus S_0 =\emptyset$, or there exists a $\{x,y\}\in\calt$ such that $SP(x,y)\setminus S_0 =\emptyset$, then no $M$ compatible solution exists and we make the next guess of $M$.
\paragraph{}
We now move on to the construction of an instance of {\hs}. Let $\calf_1=\{SP(u,v)\setminus S_0 \mid u,v \in C \ \land \ M[u,v]=1\}$. 
Let $T_1=\{\{u,v\} \mid \{u,v\}\in T \land \exists f\in\calf_1 (f \subseteq SP(u,v))\}$.
Then every hitting set for $\calf_1$ also hits $SP(u,v)$ for all terminal pairs $\{u,v\} \in T_1$.
Thus, we now consider $T_0=T \setminus T_1$.
We also note that the set $T_0$ can be computed in polynomial time. We now move on to construct a family $\calf_2$ which will be part of the instance of {\hs}.

\begin{rr}\label{rr:2}
    Let $T_0^*=T_0$, we initialize a family $\calf_2=\emptyset$.
While there exists a pair $\{u,v\} \in T_0^*$
such that $u \in I,v \in C_0$, do the following:
\begin{itemize}
    \item $\calf_2=\calf_2\cup \{u\}$;
    \item Let $T_0^*=T_0^* \setminus \{\{x,y\}\mid \{x,y\}\in T_0 \ \land \ u \in SP(x,y)$\}.
\end{itemize}
\end{rr}

 \paragraph{}

The above reduction is sound, since $\{u,v\}\in T_0$ and $M[w,v]=0$ for every neighbor $w$ of $u$ which belongs to the shortest path between $u$ and $v$, and $SP(u,v)\setminus S_0 \neq \emptyset$, it must hold that $SP(u,v)\setminus S_0 = \{u\}$. And if any hitting set hits $\{u\}$ ($u$ must be in it), then it must hit $SP(x,y)$ in which $u$ belong.
\paragraph{}

Let $T_0^*$ and $\calf_2$ be the resulting set of terminal pairs and set family respectively. We claim that for a pair $\{u,v\} \in T_0^*$ such that $u \in I,v \in C_0$, $SP(u,v)\setminus S_0$ equals $\{u,v\}\setminus S_0$. This is because since $\{u,v\}\in T_0$ we have that $M[x,y]=0$ for every $x,y$ such that of $x$ is neighbor $u$ and  $y$ is neighbor of $v$ and $x,y$ belongs to a same shortest path between $u$ and $v$. Hence, $SP(u,v)\setminus S_0=\{u,v\}\setminus S_0$. 
To this end, we construct $\calf_2=\calf_2\cup (\{u,v\}\setminus S_0)$ and $\calf=\calf_1 \cup \calf_2$.

\begin{claimn}
    $I'= (\bigcup_{f\in\calf}f, \calf, k)$  is a yes instance of {\hs} if and only if there is a $M$-compatible terminal monitoring set for $\calt$ of size at most $k$
\end{claimn}
\paragraph{}

The proof of the above claim easily follows from the construction of $\calf_1$ and $\calf_2$.
Further, by applying Observation \ref{obs:vertexCover} to $(\bigcup_{f\in\calf}f, \calf, k)$, we can obtain an equivalent instance $(\bigcup_{f\in\calf_1\cup\calf_3}f, \calf_1\cup \calf_3, k)$ of {\hs}.
We solve the instance $(\bigcup_{f\in\calf_1\cup\calf_3}f, \calf_1\cup \calf_3, k)$; since $|\calf_1 \cup \calf_3| \leq t^2+O(k^2)=O(t^2)$, we can solve this instance in time FPT in $t$.
The correctness of algorithm for weighted-{\tms} follows from exhaustive guesses on $M$, further there can be at most $2^{t^2}$ distinct guesses of $M$ and thus we can solve weighted-{\tms} in time FPT in $t$.
This finishes the proof of Theorem \ref{thm:vc}.

\subsection{Feedback Edge Number: Proof of Theorem \ref{thm:FEN}}

Jansen \cite{Jansen17}  gave an algorithm (which we refer to as Jansen's algorithm) running in time FPT by feedback edge number of the input graph to solve {\sc{hitting paths in a graph}}. In {\tms}, we need to hit a set of connected subgraphs, and we found that Jansen's algorithm solves {\tms} correctly, and for the proof sketch of Theorem \ref{thm:FEN}, we recall Jansen's algorithm from \cite{Jansen17} and discuss it in this section.
\paragraph{}

We may assume that input graph $G$ is connected. Otherwise, we create a clique on $4$ vertices $Z=\{z_1,z_2,z_3,z_4\}$ and connect this clique to $G$ by adding an arbitrary edge. This will increase feedback edge number of $G$ by a constant, and it will not change the solution, as no new shortest path between any terminal pair introduced. In the rest of the section we assume that feedback edge number of $G$ is $t$. Given an instance $(G,\calt,k)$ of {\tms}, the following preprocessing is performed.

\textbf{Preprocessing 1 (adapted from Observation 3 in \cite{Jansen17}})\label{preprocessing}: While there is a vertex $v \in G$ of degree one.

\begin{itemize}
   
    \item If there is a  terminal pair $\{v\}$ in $\calt$, then put $v$ in solution $S$, decrease $k$ by one, and remove every terminal pair containing $v$ from $\calt$. Otherwise, we replace every terminal pair $\{v,y\}$ with $\{u,y\}$ in $T$ where $u$ is the only neighbor of $v$ in $G$. Remove $v$ from $G$.
\end{itemize}
\paragraph{}

The above preprocessing is safe, if both vertices of a pair are $v$, then $v$ must be in the solution, else $v$ can be replaced by $u$ in any solution that contains $v$.

\begin{figure}

    \centering
    \begin{tikzpicture}[scale=0.9]


    \node [ellipse, minimum height=1.5cm,minimum width= 4cm, label ={0:\footnotesize{\text{$G$}}}] (z) at (2*3,-2*1) {};
   
   \node[shape=circle, draw, fill = black, scale= 0.35, font=\footnotesize] ({21}) at (2*3+0.4,-2*1+0.3){};
    \node[shape=circle, draw, fill = black, scale= 0.35, font=\footnotesize] ({22}) at (2*3-0.4,-2*1+0.3){};
    \node[shape=circle, draw, fill = black, scale= 0.35, font=\footnotesize] ({23}) at (2*3+0.4,-2*1-0.3){};
    \node[shape=circle, draw, fill = black, scale= 0.35, font=\footnotesize] ({24}) at (2*3-0.4,-2*1-0.3){};
     \node[shape=circle, draw, fill = black, scale= 0.35, font=\footnotesize] ({25}) at (2*3-0.75,-2*1-0){};
     \node[shape=circle, draw, fill = black, scale= 0.35, font=\footnotesize] ({26}) at (2*3+0.75,-2*1-0){};
   
   \draw[thin,black!40] ({21}) to ({22});
   \draw[thin,black!40] ({21}) to ({23});
   \draw[thin,black!40] ({21}) to ({24});
   \draw[thin,black!40] ({21}) to ({26});
   \draw[thin,black!40] ({23}) to ({26});
   \draw[thin,black!40] ({22}) to ({23});
   \draw[thin,black!40] ({22}) to ({24});
   \draw[thin,black!40] ({23}) to ({24});
   \draw[thin,black!40] ({22}) to ({25});
   \draw[thin,black!40] ({24}) to ({25});
   
   \foreach \i in {1,2,3,4}{
     
            \node[shape=circle, draw, fill = black, scale= 0.15, font=\footnotesize,label={}] ({21\i}) at (2*3-0.8+\i/3,-2*1+0.8){};
   }
   \foreach \i in {1,2,3}{
     
            \node[shape=circle, draw, fill = black, scale= 0.15, font=\footnotesize,label={}] ({22\i}) at (2*3+1.5,-2*1-0.6+\i/3){};
   }
   
    \foreach \i in {1,2,3}{
     
            \node[shape=circle, draw, fill = black, scale= 0.15, font=\footnotesize,label={}] ({23\i}) at (2*3-1.5,-2*1-0.6+\i/3){};
   }
   \draw[thin,black!40] ({211}) to ({214});
   \draw[thin,black!40] ({221}) to ({223});
   \draw[thin,black!40] ({231}) to ({233});
   
   \draw[thin,black!40,bend right =30] ({211}) to ({22});
   \draw[thin,black!40,bend left =30] ({214}) to ({21});
    \draw[thin,black!40,bend right =30] ({223}) to ({26});
   \draw[thin,black!40,bend left =30] ({221}) to ({26});
    \draw[thin,black!40,bend right =30] ({231}) to ({25});
   \draw[thin,black!40,bend left =30] ({233}) to ({25});

\end{tikzpicture}
\caption{$G$ with minimum degree two, darkened vertices forms $V_{\geq 3}$. }\label{fig:structures}
\end{figure}
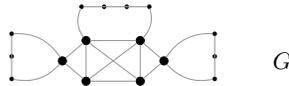

\paragraph{}

After Preprocessing 1, $G$ has minimum degree two, and we will assume that for every $I=(G,\calt,k)$, $G$ has minimum degree at least two. 
 For a graph $G$ with degree at least two, let $V_{\geq3}$ be the set of all the vertices of $G$ with degree at least three (Figure  \ref{fig:structures}), $G-V_{\geq3}$ is a disjoint union of paths as $V_{\geq3}\neq \emptyset$ (we added a clique on $Z$ in $G$) and $G$ is connected. Let $\cal D$ be the set of all the components (paths) in $G-V_{\geq3}$. Observe that every component of $\cal D$ is connected to rest of the graph by $2$ edges, each of which connecting an endpoint of $D$ to a vertex of $V_{\geq 3}$ as every vertex in $D$ has degree two in $G$ (see also \cite{Jansen17}). In the remaining part, for the graph in context, we simply use $V_{\geq3}$ and $\cal D$ for it as defined above. Given an instance $I=(G,\calt,k)$ of {\tms}, the following holds.
\begin{observation}\label{Obs:subgraphproperty1}
     For every $\{x,y\}\in \calt$, $SP(x,y)$ induces a connected graph and every vertex in $SP(x,y)\setminus \{x,y\}$ has degree at least two in $G[SP(x,y)]$.  
\end{observation}
\begin{proof}
    It follows from the fact that a union of all the shortest path between $(x,y)$ is a subgraph of graph induced by $SP(x,y)$.
\end{proof}

\begin{observation} \label{obs:subpath}
    For every $\{x,y\} \in {\calt}$ and $D\in \cal D$ it holds that: if $SP(x,y)\subseteq V(D)$, then $SP(x,y)$ induces a sub path of $D$ in $G$.
\end{observation} 
\paragraph{}

 Given an instance $(G,\calt,k)$ of {\tms}, we define the following (using similar notations to \cite{Jansen17}). For every $D\in \cal D$, let {\sc{opt}}($D$) be the minimum size of a terminal monitoring set for $\{\{x,y\} \mid \{x,y\}\in \calt \land SP(x,y)\subseteq V(D)\}$. From Observation \ref{obs:subpath} and Lemma \ref{lemma: folklore}, we have that {\sc{opt}}$(D)$ can be computed in polynomial time for every $D\in \cal D$.

\begin{lemma}[Lemma 5 in {\cite{Jansen17}}, stating for TMS]\label{lemma:opt(D)oropt(D)+1}
  Given an instance $(G,\calt,k)$ of {\tms}, there exists a minimum size terminal monitoring set $S'$ for $T$ such that: for every $D\in \cal D$, $S'$ contains either {\sc{opt}}($D$) or {\sc{opt}}($D$)$+1$ vertices of $D$.
\end{lemma}

\paragraph{}

Proof for Lemma \ref{lemma:opt(D)oropt(D)+1} uses same arguments as given in the proof of \cite[Lemma 5]{Jansen17} and only requires supplementing them with Observation \ref{Obs:subgraphproperty1} and Observation \ref{obs:subpath}, hence we skip it.

\paragraph{}

Jansen's algorithm \cite{Jansen17} while solving an instance of {\hpg} makes successive guesses (to branch), and if it decides that a guess may lead to a solution within the budget, it constructs an instance of {\hpfbs}, and solves it in polynomial time. We found that Jansen's algorithm correctly solves {\tms} as well, crucially by correctly constructing instances of {\hpfbs} for input $I=(G,\calt,k)$. For completeness, we recall steps of Jansen's algorithm from \cite{Jansen17} which are divided into branching and construction of {\hpfbs}, we demonstrate them for the input instance $I=(G,\calt,k)$ of {\tms}. 

\subsubsection{\bf{Branching} ({Section 3.2 in \cite{Jansen17}}):} First make a guess on the vertices of $V_{\geq3}$ which will be in the solution, and a guess on the number of vertices in the solution from every $D\in \cal D$ using the functions $f_v:V_{\geq3}\to \{0,1\}$  and $f_d:{\cal D}\to \{0,1\}$. If $k\geq\sum_{v\in V_{\geq3}}f_v(v)$+ $\sum_{D\in {\cal D}}({\textsc{opt}}(D)+f_d(D))$, then construct an instance of {\hpfbs} . 

\subsubsection{\bf{Construction of an instance of {\hpfbs}} ({Section 3.2 in \cite{Jansen17}}):}
Given $f_v,f_d$. Create $G_1$ and ${\cal D}_1$ as copies of $G$ and $\cal D$ respectively, and $f_d$ remains same for ${\cal D}_1$. Let ${\cal U}_1= \{SP(x,y)\mid \{x,y\}\in \calt\}$, and do the following.
\begin{itemize}
    \item For every $U\in {\cal U}_1$ : if $U$ contains a vertex $v\in V_{\geq3}$ such that $f_v(v)=1$ or $U$ contains all the vertices of a $D\in {\cal D}$ such that $({\textsc{opt}}(D)+f_d(D))>0$, then remove $U$ form ${\cal U}_1$.
    \item For every $D\in {\cal D}_1$ such that $({\textsc{opt}}(D)+f_d(D))=0$ :  remove $D$ from ${\cal D}_1$, remove $V(D)$ from $G_1$, and remove $V(D)$ from every set $U\in {\cal U}_1 $.
    \item  For every $v\in V_{\geq3}$ such that $f_v(v)=1$ : remove $v$ from $G_1$.
    \item  Contract remaining vertices of $V_{\geq3}$ in $G_1$ into a single vertex $z$ in $G_1$, and replace $U\cap V_{\geq3}$ (if non empty) with $z$ in every $U\in {\cal U}_1$. 
    \item Assign a distinct number in $[|{\cal D}_1|]$ to every remaining path in ${\cal D}_1$, set $b(i)=({\textsc{opt}}(D)+f_d(D_i))$, where $D_i\in {\cal D}_1$. Output  $(G_1,z,{\cal U}_1,b)$.
\end{itemize}

\paragraph{}

For the input instance $I=(G,\calt,k)$ of {\tms}, for every guess $f_v,f_d$, and  correspondingly constructed $(G_1,z,{\cal U}_1,b)$, the following claims hold.

\begin{claimn}[Claim 5 in {\cite{Jansen17}},  stating for $(G,\calt,k)$]\label{claim:itisflower}
    $G_1$ is a flower graph with core $z$ and every vertex set in ${\cal U}_1$ induces a simple path in $G_1$.
\end{claimn}

\begin{claimn}[{Claim 6 in \cite{Jansen17}}, stating for $(G,\calt,k)$]\label{claim:flowercorrect}
     The following two statements are equivalent.
    \begin{itemize}
        \item There exists a terminal monitoring set $S$ for $\calt$ such that : $|S\cap V(D)|={\textsc{opt}}(D)+f_d(D)$ for every $D\in \cal D$, and $v\in S \Leftrightarrow f_v(v)=1$ for every $v\in V_{\geq3}$.
    
        \item  There exists a solution for instance $(G_1,z,{\cal U}_1,b)$ of HPFB.
    \end{itemize}
\end{claimn}
\paragraph{}

Proof of Claim \ref{claim:itisflower} and proof of Claim \ref{claim:flowercorrect} for instance $(G,\calt,k)$ follow from the same arguments given in proof of {\cite[Claim 5]{Jansen17}} and proof of {\cite[Claim 6]{Jansen17}} respectively by supplementing them with Observation \ref{Obs:subgraphproperty1} and Observation \ref{obs:subpath}, hence we skip them.
\paragraph{}

If any constructed instance of {\hpfbs} has a solution, then Jansen's algorithm returns YES, otherwise NO. Further, correctness on $(G,\calt,k)$ follows from Claim \ref{claim:flowercorrect}, Claim \ref{claim:itisflower}, and Lemma \ref{lemma:opt(D)oropt(D)+1}. Size of $V_{\geq3}$ and $\cal D$ can be bounded by $2t$ and $3t$ respectively if  $V_{\geq3}\neq \emptyset$ \cite{Jansen17}. Further, Jansen's algorithm makes at most $2^{5t}$ guesses, and its running time is bounded by $2^{5t}\cdot (|V(G)|+|{\cal U}|)^{O(1)}$ \cite{Jansen17}.
\section{Hardness Results}

 In this section we prove the following theorems.
\begin{theorem}\label{thm:hardness-one}
(a) {\tms} is NP-hard.

(b) {\tms} is W[2]-hard with respect to solution size.
\end{theorem}

\begin{theorem}\label{thm:hardness-fvs}
For every fixed $0< \alpha\leq 0.5$,  $\alpha$-RTMS is W[1]-hard with respect to feedback vertex number of the input graph plus solution size, even when the input graph is a planar graph.
\end{theorem}

\subsection{Proof of Theorem \ref{thm:hardness-one}}

We reduce {\sc Red-Blue Dominating Set (RBDS)} to {\tms}. In  {\sc RBDS } we are given a bipartite graph $G=(V_B\cup V_R,E)$, and it is asked if there is a vertex set $D\subseteq V_B$ of size at most $k$ such that every vertex in $V_R$ is adjacent to at least one vertex in $D$. It is known that {\sc RBDS } is W[2]-hard parameterized by solution size $k$ (see \cite{DBLP:books/sp/CyganFKLMPPS15}).
\paragraph{}

Let $I=(G=(V_B\cup V_R,E),k)$ be the input instance of {\sc RBDS }, let $V_R= \{r_1,r_2,\dots r_n\}$, and let $V_B= \{b_1,b_2,\dots b_m\}$. We construct an instance $I'$ of {\tms} as follows. We construct the graph $G'$ as follows. Create a vertex set $V'_R= \{r'_1, r'_2,\dots r'_n\}$. For every $i\in[n]$, we connect $r'_i$ to all the neighbors of $r_i$ in $V_B$, and call the set of all these introduced edges as $E'$. Essentially we are creating a twin vertex for every $r\in V_R$.  We construct terminal set $T=\{\{r_i,r'_i\}\mid i\in [n]\}$. The instance $I'=(G'=(V_B\cup V_R\cup V'_R, E\cup E'), T, k)$.

\begin{lemma}\label{lemma:equivalenceRBDM}
    $I$ is a yes instance of {\sc RBDS } if and only if $I'$ is a yes instance of {\tms}.
\end{lemma}
    
\begin{proof}
    For the forward direction, let $D\subseteq V_B$ be a solution to $I$, in this case we can verify that $D$ is also a solution to $I'$, since every terminal pair in $T$ are twins and are not adjacent, they have the same neighborhood in $V_B$, and the distance between them is $2$, and if a vertex in $D$ is adjacent to an $r_i\in V_B$, then it is adjacent to its twin $r'_i\in R'_B$ as well.
    \paragraph{}

    For the other direction, let $S$ be a solution to $I'$, in this case if $S$ contain a vertex from $r_i\in V_R\cup V'_R$, then we can observe that $r_i$ is only hitting the terminal pair $(r_i,r'-i)$, because it doesn't belong to a shortest path between any other terminal pair. So we can safely replace $r_i$ with any of its neighbor in $V_B$. This way while $S$ contains a vertex from $V_R\cup V'_R$, we replace that vertex with any of its neighbor in $V_B$. Now $S$ contains only vertices of $V_B$ and since its a set that hit all terminal pairs in $T$, it is easy to see that it dominates all the vertices in $V_R$.
    \end{proof}

\subsection{Proof of Theorem \ref{thm:hardness-fvs}}

\begin{figure}[ht]

   \usetikzlibrary { decorations.pathmorphing, decorations.pathreplacing, decorations.shapes, } 

    \centering
    \begin{tikzpicture}[scale=0.6]

 \foreach \j in {1,2,4,5}{
       
            \node[shape=circle, draw, fill=black, scale= 0.5, font=\small,label={-10}:{\small{$u'_{i,\j}$}}] ({1\j}) at (-2-\j,4-\j){};
   }

\foreach \j in {1,2,4,5}{
       
            \node[shape=circle, draw, fill=black, scale= 0.5, font=\small,label={190}:{\small{$u_{i,\j}$}}] ({2\j}) at (7-\j,-2+\j){};
   }

 \foreach \j in {3}{
       
            \node[shape=circle, draw, fill=black, scale= 0.5, font=\small,label={-10}:{\small{$u'_{i,\j}$}}] ({1\j}) at (-2-\j,4-\j){};
   }

\foreach \j in {3}{
       
            \node[shape=circle, draw, fill=black, scale= 0.5, font=\small,label={210}:{\small{$u_{i,\j}$}}] ({2\j}) at (7-\j,-2+\j){};
   }

 \node[shape=circle, draw, fill=black, scale= 0.5, font=\small,label=90:{{$b$}}] (b) at (-0.5,5){};

\node[shape=circle, draw, fill=black, scale= 0.5, font=\small,label=90:{{$z_i$}}] (zi) at (-0.5,3){};
\node[shape=circle, draw, fill=black, scale= 0.5, font=\small,label=90:{{$z'_i$}}] (z'i) at (-0.5,-1){};

 \draw[thin,black]({11}) to ({15});
\draw[thin,black]({21}) to ({25});

[decoration=snake,
   line around/.style={decoration={pre length=1,post length=1}}]

\draw[ thin,black]({11}) to (25);
\draw[thin,black]({15}) to (21);

\foreach \i in {1,2}{     
 \draw[decorate,decoration=zigzag, thin,black,bend right =120-20*\i,decorate]({2\i}) to (b);
}
\foreach \i in {3}{      
\draw[decorate,decoration=zigzag,thin,black,bend right =120-20*\i] ({2\i}) to (b);

}
\foreach \i in {4,5}{      
\draw[decorate,decoration=zigzag,thin,black,bend right =120-20*\i] ({2\i}) to (b);

}
\foreach \i in {1,2}{      
\draw[decorate,decoration=zigzag, thin,black!80,bend left =\i*30-30] ({1\i}) to (b);
}
\foreach \i in {3}{      
\draw[decorate,decoration=zigzag,  thin,black,bend left =\i*25-30] ({1\i}) to (b);

}
\foreach \i in {4,5}{      
\draw[decorate,decoration=zigzag, thin,black,bend left =\i*25-30] ({1\i}) to (b);
}

            \node[shape=circle, draw, fill=black, scale= 0.5, font=\small,label={$p_i$}] (p) at (0,1){};
            
\draw[decorate,decoration=saw, thick,red] (p) to (23);

 \node[shape=circle, scale= 0.5, font=\small,label={$P(u'_{i,5},b)$}] (lb) at (-8.4,1.5){};
\node[shape=circle, scale= 0.5, font=\small,label={}] (lb) at (-4.5,1.1){};
\node[shape=circle, scale= 0.5, font=\small,label={}] (lb) at (-4.5,2.1){};
\node[shape=circle, scale= 0.5, font=\small,label={$L-1$}] (lb) at (-6.5,3.3){};
\node[shape=circle, scale= 0.5, font=\small,label={}] (lb) at (-4.5,-1.1){};
\node[shape=circle, scale= 0.5, font=\small,label={}] (lb) at (-4.5,-2.1){};
\node[shape=circle, scale= 0.5, font=\small,label={}] (lb) at (-4.5,-2.9){};

 \node[shape=circle, scale= 0.5, font=\small,label={$P(u_{i,1},b)$}] (lb) at (8,1.5){};
\node[shape=circle, scale= 0.5, font=\small,label={}] (lb) at (4.5,1.1){};
\node[shape=circle, scale= 0.5, font=\small,label={}] (lb) at (4.5,2.1){};
\node[shape=circle, scale= 0.5, font=\small,label={$L-1$}] (lb) at (5.5,3.4){};
\node[shape=circle, scale= 0.5, font=\small,label={}] (lb) at (4.5,-1.1){};
\node[shape=circle, scale= 0.5, font=\small,label={}] (lb) at (4.5,-2.1){};
\node[shape=circle, scale= 0.5, font=\small,label={}] (lb) at (4.5,-2.9){};

\node[shape=circle, scale= 0.5, font=\small,label={\textcolor{red}{$L_p-1$}}] (lb) at (1,0.88){};

\end{tikzpicture}
\caption{An example of $H_i$ connected to bridge vertex $b$, with each path $P(u_{i,j},b)$ and $P(u'_{i,j},b)$ is of length $L$ and contains $L-1$ intermediate vertices. }\label{fig:hardstructures}
\end{figure}
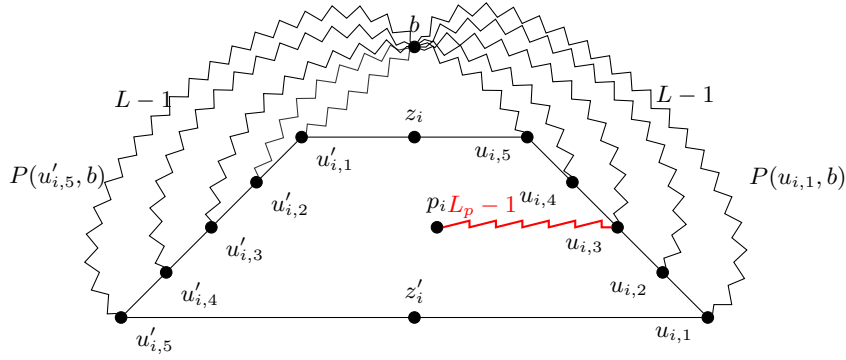

 {\mcis} is a well know problem in parameterized complexity and it is known to be  W[1]-hard when parameterized by solution size $k$ \cite{DBLP:books/sp/CyganFKLMPPS15}. We reduce {\mcis} to {\atms}. Let $I=(G,k)$ such that $V(G)$ is partitioned into $\{V_1,V_2,\dots, V_k\}$, be an instance of {\mcis}, and we need to decide if there exists an independent set $S$ containing exactly one vertex from every $V_i$, where $i\in[k]$. We may assume that for every $i\in [k]$, $|V_i|=n$, and $n$ is odd. We begin our construction by first defining the following values.
\paragraph{}

For $0\lneq\alpha \leq 0.5$, we set 
   $L=  \lceil \frac{1}{2\alpha} \cdot n \rceil$;
   $L_p= \lceil \frac{n-1}{\alpha} \rceil$.
\paragraph{}

The construction of $G'$ is as follows (Figure  \ref{fig:hardstructures}). We create a vertex $b$ and call it the bridge vertex. For every vertex set $V_i$ in $G$ we construct a gadget $H_i$ as follows.
\begin{itemize}
    \item Create a set $U'_i= \{u'_{i,j}\mid v_{i,j}\in V_i\}$ and make a path on these vertices as $(u'_{i,1},u'_{i,2},\dots u'_{i,n})$. Create another set  $U_i= \{u_{i,j}\mid v_{i,j}\in V_i\}$ and create a path $(u_{i,1},u_{i,2},\dots u_{i,n})$ as well. Connect $u'_{i,1}$ to $u_{i,n}$ by creating a vertex $z_i$ and connect $z_i$ to $u'_{i,1}$ and $u_{i,n}$. Similarly, connect $u'_{i,n}$ to $u_{i,1}$ by creating a vertex $z'_i$.
    \item create a vertex $p_i$, create a path with $L_p-1$ vertices, and connect one endpoint of this path to $p_i$ and another endpoint with vertex $u_{i,(n+1)/2}$, We denote this path from $p_i$ to $u_{i,(n+1)/2}$ by $P(p_i,u_{i,(n+1)/2})$ which include both $p_i$ and $u_{i,(n+1)/2}$. 
\end{itemize}

 For every $i\in[k]$, we connect the vertices of $H_i$ to $b$ as follows.
\begin{itemize}
    \item For every $i\in[k]$ and $j\in [n]$, connect $u_{i,j}$ to $b$ by creating a path with $L-1$ intermediate vertices, denote this path $P(u_{i,j},b)$. Thus, $d(u_{i,j},b)=L$.
     Similarly, connect $u'_{i,j}$  with vertex $b$ by  creating a path with $L-1$ intermediate vertices, and denote this path by $P(u'_{i,j},b)$. Thus, $d(u'_{i,j},b)=L$.
\end{itemize}

We now construct the terminal set $\cal T$ as follows.
\begin{itemize}
    \item  From every $H_i$ we put $\{p_i, u_{i,(n+1)/2}\}$ in $\cal T$. It is a vertex selection constraints.  
    \item Let $E_{i,j}$ be the set of edges with one endpoint in $V_i$ and another in $V_j$ in $G$, for every $i,j\in[k]$, such that $i<j$, construct $T_{i,j}=\{\{u'_{i,i'},u'_{j,j'}\}\mid i',j'\in[n] \land v_{i,i'}v_{j,j'}\in E_{i,j}\}$ and put all the pairs of $T_{i,j}$ in $\calt$.  Addition of these terminal pairs acts as edge verification constraints.
\end{itemize}
\paragraph{}

The instance $I'=(G',{\cal T},k)$. We can see that if we remove vertices $b$ and $u_{i,1}$ for every $i\in[k]$, then $G'$ becomes acyclic, thus the feedback vertex number (FVN) of $G'$ is $k+1$, thus the parameter FVN + solution size remain $O(k)$. Further, it is easy to verify that constructed graph is a planar graph as no crossing edge is constructed (see Figure \ref{fig:hardstructures}). We now move on to show the correctness of our reduction. We start with stating some observation which will help establishing equivalence. 

\begin{observation}\label{obs:shortest_distance_between_two_hi}
    For every distinct $i,j\in [k]$, for every $x\in U'_i\cup U_i$ and every $y \in U'_j\cup U_j$, the distance $d(x,y)= 2L$.
\end{observation}

\begin{definition}\label{def:Hirepresentative}
    For an $i\in[k]$, we say a vertex $w\in V(G')$ is $H_i$-representative if the following holds.
    \begin{itemize}
        \item $w\in V(P(p_i,u_{i,(n+1)/2}))$ (OR) $d(u_{i,(n+1)/2},w)\leq \frac{n-1}{2}$.
    \end{itemize}
\end{definition}

\begin{observation}\label{obs:Hi-rep-can't_be-of-two}
   No vertex in $V(G')$ can be $H_i$-representative for two or more distinct $i\in[k]$.
\end{observation}
\begin{proof}
    This is because distance between $u_{i,(n+1)/2}$ and $u_{j,(n+1)/2}$ for every distinct $i,j\in[k]$ is $2L$ and $L$ is at least $n$.
  
\end{proof}

\begin{lemma}\label{lemma:Himustbethere}
    If $S$ is a terminal monitoring set of size at most $k$ for the instance $(G',{\cal{T}},k)$, then for every $i\in[k]$, there is a vertex $w\in S$ such that $w$ is an $H_i$-representative. 
\end{lemma}
\begin{proof}
    Recall that from every $H_i$ we put the pair $\{p_i,u_{i,{n+1\over 2}}\}$ in $\cal{T}$. 
    To satisfy this constraint there must be a vertex $w\in S$ such that $d(u_{i,{n+1\over 2}},w) + d(w,p_i)\leq (1+\alpha)\cdot d(p_i,u_{i,{n+1\over 2}})$. If $w$ belongs to $P(p_i,u_{i,{n+1\over 2}})$, then the constraint is trivially satisfied. Otherwise, the path from $p_i$ to $w$ must contain $u_{i,{n+1\over 2}}$, and then
    \begin{align*}
        d(u_{i,{n+1\over 2}},w) + d(w,p_i)&=  d(u_{i,{n+1\over 2}},w)+ d(u_{i,{n+1\over 2}},w)+  d(u_{i,{n+1\over 2}},p_i).\\
        &= d(u_{i,{n+1\over 2}},w)+ d(u_{i,{n+1\over 2}},w)+  L_p.\\
        \text{to satisfy the constraint}\\
        (1+\alpha)\cdot L_p & \geq d(u_{i,{n+1\over 2}},w)+ d(u_{i,{n+1\over 2}},w)+  L_p.\\
        (1+\alpha)\cdot L_p -L_p & \geq 2\cdot d(u_{i,{n+1\over 2}},w).\\
        \alpha\cdot L_p & \geq 2\cdot d(u_{i,{n+1\over 2}},w).\\
         \alpha\cdot \lceil{{n-1}\over {\alpha} }\rceil & \geq 2\cdot d(u_{i,{n+1\over 2}},w).\\
         \alpha\cdot (\frac{n-1}{\alpha}+\epsilon) &\geq 2\cdot d(u_{i,{n+1\over 2}},w).\\
         {{n-1}}+ \alpha \cdot \epsilon &\geq 2\cdot d(u_{i,{n+1\over 2}},w).\\
         \text{since $\epsilon$ is rounding off for ceil } & \text{ function and must be strictly less than one, and}\\
         \text{ $\alpha$ is at most $0.5$, we have}\\
         {{n-1}\over 2}+ 0.25 &\geq d(u_{i,{n+1\over 2}},w).
    \end{align*}
    Since distances are in integers, $d(u_{i,{n+1\over 2}},w)$ must be at most ${{n-1}\over 2}$.
    \end{proof}

\begin{definition}\label{def:nice}
A solution set $S$ of $I'$ is {\textit nice } if and only if the following holds.
\begin{itemize}
    \item $|S\cap U_i|=1$ for every $i\in[k]$; (AND) $S= \bigcup_{i\in[k]}S\cap U_i  $.
\end{itemize}

\end{definition}

\begin{lemma}\label{lemma:Hinotssatisfy}
   For every distinct $i,j,l\in[k]$, such that $j<l$, the following holds: if a vertex $w$ is an $H_i$-representative, then  for every $\{x,y\}\in T_{j,l}$, $d(x,w)+d(w,y)\gneq (1+\alpha)\cdot d(x,y)$, that is $w$ satisfies constraint impose by no terminal pair in $T_{j,l}$.
\end{lemma}

\begin{proof}
    Since $i,j,l$ are distinct, and by definition $H_i$ representative vertex $w$ is at most $n-1\over 2$ distance apart from $u_{i,{n+1\over 2}}$ or belong to $P(p_i,u_{i,{n+1\over 2}})$, in both cases $w$ is at least $L-{n-1\over 2}$ distance apart from  vertex $b$ (as $u_{i,{n+1\over 2}}$ is at least $L$ distance apart from $b$), in this case to reach $w$ from every vertex of $H_j$ and $H_l$ one has to go via $b$. And thus, for any pair $\{x,y\}\in T_{j,l}$, the distance $d(w,x)+d(w,y)$ is 

    \begin{align*}
        d(w,x)+d(w,y) &= d(w,b)+d(b,x) +d(w,b)+d(b,y).\\
                     &\geq L-{n-1\over 2} + L +L-{n-1\over 2} + L.\\
                     &\geq 4L - 2\cdot {n-1\over 2}.\\
                     & \geq 4L -n +1.\\
                     &\geq 3L+1.\text{ (As $L$ is at least $n$)}.\\
                     &\geq (1+0.5)\cdot 2L +1.\\
                     &\geq (1+\alpha)\cdot 2L+1.\\
                     &\gneq (1+\alpha)\cdot d(x,y).\\           
    \end{align*}
\end{proof}

Now consider the following observation.

\begin{observation}\label{obs:where-Hi-rep-can_be}
   For an $i\in[k]$, if a vertex $w$ is $H_i$-representative, then one of the following holds.
   \begin{itemize}
       \item  $w$ belongs to $V(P(p_i,u_{i,{n+1\over 2}}))$.
       \item There exists a $j$ such that $1\leq j \leq {n-1 \over 2}$ and $w$ belongs to the first $j$ vertices of constructed path $P(u_{i,j},b)$ starting from $u_{i,j}$.
       \item There exists a $j$ such that $ {n+1 \over 2}\leq j \leq n$ and $w$ belongs to the first $n-j+1$ vertices of constructed path $P(u_{i,j},b)$ starting from $u_{i,j}$.
   \end{itemize}
\end{observation}
\begin{proof}
    We refer to Figure \ref{fig:hardstructures} to verify that a vertex $w$ with $d(u_{i,{n+1\over 2}},w)\leq {n-1\over 2}$ belongs to vertex sets as  mentioned above.
\end{proof}

\begin{lemma}\label{lemma:nice}
    If there is a solution $S$ for instance  $(G',{\cal T}, k)$, then there exists a nice solution $S^*$ for instance  $(G',{\cal T}, k)$.
\end{lemma}

\begin{proof}
    
    Recalling Lemma \ref{lemma:Himustbethere} , For every $i\in[k]$,  $S$ must contain a vertex $w$ which is an $H_i$-representative. Further, recalling Observation \ref{obs:where-Hi-rep-can_be}, $w$ must belong to path $P(p_i,u_{i,{n+1\over2}})$, or to $P(u_{i,i'},b)$ for an $i'\in[n]$.
    Initially we make $S^*$ as a copy of $S$. For every $i\in [k]$ we do the following.
    \begin{itemize}
        \item If $H_i$-representative vertex $w$ belongs to path $P(p_i,u_{i,{n+1\over2}})$, then we replace $w$ with $u_{i,{n+1\over2}}$ in $S^*$.
        \item Else, if $H_i$-representative vertex $w$ belongs to path $P(u_{i,q},b)$ for a $q\in[n]$, then we replace $w$ with $u_{i,q}$ in $S^*$.
    \end{itemize}
    \paragraph{}

    The above replacements are safe, we argue this for a fixed $i\in[k]$, and the similar arguments hold for every $i\in[k]$. Recalling Lemma \ref{lemma:Hinotssatisfy}, $H_i$-representative $w$ can only satisfy constraint imposed by terminal pairs of $T_{i,j}$ where $i<j<k$ or $T_{j,i}$ where $1<j<i$. We will argue for a fixed $T_{i,j}$ where $i<j<k$ and similar arguments hold for all $T_{i,j}$ where $i<j<k$ and $T_{j,i}$ where $1<j<i$. By construction, every pair in $T_{i,j}$ has one vertex from $U'_i$ and other vertex from $U'_j$. Let $\{u'_{i,i'},u'_{j,j'}\}$ be a pair whose constraint is satisfied by $w$, that is $d(u'_{i,i'},w)+d(w,u'_{j,j'}) \leq (1+\alpha)\cdot d(u'_{j,j'},u'_{i,i'})$, and let $u_{i,q}$ be the replacement of $w$ in the above. Let $u_{i,q}$ be at distance $l$ from $w$, consider the following cases.\\
    \textbf{Case 1:} $w$ belongs to $P(p_i,u_{i,{n+1\over2}})$. In this case, $d(u'_{i,i'},u_{i,q})= d(u'_{i,i'},w)-l$ and $d(u_{i,q},b)=d(w,b)-l$. Further, the shortest paths from $w$ and $u_{i,q}$ to $u'_{j,j'}$ go via $b$, so we have that 
    \begin{align*}
        d(u'_{i,i'},w)+d(w,u'_{j,j'}) &= d(u'_{i,i'},w)+d(w,b)+d(b,u'_{j,j'}).\\
                     &= d(u'_{i,i'},u'_{i,q})+l+d(u_{i,q},b)+l+d(b,u'_{j,j'}).\\
                     &= d(u'_{i,i'},u'_{i,q})+d(u_{i,q},u'_{j,j'})+2l.\\      
    \end{align*}
    \textbf{Case 2:} $w$ belongs to $P(u_{i,q},b)$. In this case, $d(u'_{i,i'},u_{i,q})= d(u'_{i,i'},w)-l$ and $d(u_{i,q},b)=d(w,b)+l$. Further, the shortest paths from $w$ and $u_{i,q}$ to $u'_{j,j'}$ go via $b$, so we have that 
    \begin{align*}
        d(u'_{i,i'},w)+d(w,u'_{j,j'}) &= d(u'_{i,i'},w)+d(w,b)+d(b,u'_{j,j'}).\\
                     &= d(u'_{i,i'},u'_{i,q})+l+d(u_{i,q},b)-l+d(b,u'_{j,j'}).\\
                     &= d(u'_{i,i'},u'_{i,q})+d(u_{i,q},u'_{j,j'}).\\      
    \end{align*}

    In both the above cases $u_{i,q}$ also satisfies the constraint imposed by pair $\{u'_{i,i'},u'_{j,j'}\}$. Further, since solution $S$ contains an $H_i$-representative for every $i\in[k]$, and by Definition \ref{def:Hirepresentative}, it is easy to observe that no  vertex can be $H_i$-representative of two or more distinct $i\in[k]$, we have that the above construction of $S^*$ satisfies properties of Definition \ref{def:nice}. 
   \end{proof}

\begin{lemma}\label{lemma:notsatisfy}
    For every distinct $i,j\in [k]$, and for every $i',j'\in[n]$, the following holds:
    \begin{itemize}
        \item $d(u'_{i,i'}, u_{i,i'}) + d(u_{i,i'}, u'_{j,j'})\gneq (1+\alpha)\cdot d(u'_{i,i'}, u'_{j,j'}) $.
        \item $d(u'_{j,j'}, u_{j,j'}) + d(u_{j,j'}, u'_{i,i'}) \gneq (1+\alpha)\cdot d(u'_{i,i'}, u'_{j,j'}) $.
    \end{itemize}
\end{lemma}
\begin{proof}
We prove the first statement, the second follows similar arguments. We have $L\geq n$ and $1\leq L$; refer Figure \ref{fig:hardstructures} to see that the distance  $d(u'_{i,i'}, u_{i,i'})$ is $1+n$ which is at most $2L$. The shortest path between $u'_{i,i'}$ and $u_{i,i'}$ will be achieved by following the path on $U'$ vertices and then path on $U$  vertices after going through either $z_i$ or $z'_i$ in the gadget $H_i$, while the path through vertex $b$ will be of length $2L$. Recalling Observation \ref{obs:shortest_distance_between_two_hi}, $d(u_{i,i'}, u'_{j,j'})= 2L$. Thus,

\begin{align*}
    d(u'_{i,i'}, u_{i,i'}) + d(u_{i,i'}, u'_{j,j'}) & = 2L +1+ n. \\
    & = 2L+ 1+  2\cdot \alpha \cdot (L -\epsilon) .\\
    & = 2L+ 1+   2\cdot \alpha \cdot L - 2\cdot \alpha \cdot \epsilon. \\
     & = ((1+\alpha)\cdot 2L )+1 -2\cdot \alpha \cdot \epsilon. 
     \\
     &\gneq (1+\alpha)\cdot 2L.
\end{align*}
In the above equations $\epsilon$ is the small addition of ceil function for computing $L$, which is strictly less than $1$.

\end{proof}

\begin{lemma}\label{lemma:satisfy}
    For every distinct $i,j\in [k]$, and for every distinct $i',i''\in[n]$ and distinct $j',j''\in[n]$, the following holds:
    \begin{itemize}
        \item $d(u'_{i,i'}, u_{i,i''}) + d(u_{i,i''}, u'_{j,j'}) \leq (1+\alpha)\cdot d(u'_{i,i'}, u'_{j,j'})$.
        \item  $d(u'_{j,j'}, u_{j,j''}) + d(u_{j,j''}, u'_{i,i'})\leq (1+\alpha)\cdot d(u'_{i,i'}, u'_{j,j'})$.
    \end{itemize}
\end{lemma}
\begin{proof}
    We prove the first statement, the second follows similar arguments. We have $L\geq n$ and $1\leq L$. Refer Figure \ref{fig:hardstructures}, the distance  $d(u'_{i,i'}, u_{i,i''})$ is $\leq n$ which is strictly less than $2L$. The shortest path will be achieved by following the path on $U'$ vertices and then path on $U$  vertices after going through either $z_i$ or $z'_i$ in the gadget $H_i$, while the path through vertex $b$ will be of length $2L$. Recalling Observation \ref{obs:shortest_distance_between_two_hi}, $d(u_{i,i'}, u'_{j,j'})= 2L$. Thus,

\begin{align*}
    d(u'_{i,i'}, u_{i,i''}) + d(u_{i,i''}, u'_{j,j'}) &\leq 2L +n. \\
    &\leq 2L+   2\cdot \alpha \cdot 2\cdot \alpha \cdot (L -\epsilon).\\
    & \leq 2L+    2\cdot \alpha \cdot L -2\cdot \alpha \cdot \epsilon. \\
     &\leq ((1+\alpha)\cdot 2L ) -2\cdot \alpha \cdot \epsilon. 
     \\
    &\leq (1+\alpha)\cdot 2L.
\end{align*}
In the above equations $\epsilon$ is the small addition for ceil function for calculation of $L$, which is strictly non negative. 
\end{proof}

\begin{lemma}\label{lemma:MIequivalence}
    $I$ is a yes instance of {\mcis} if and only if $I'$ is a yes instance of {\atms}.
\end{lemma}

    \begin{proof}
    For the forward direction let $S$ be a solution for instance $I$. We construct $S'= \{u_{i,j}\mid i\in[k]\land j\in [n]\land v_{i,j}\in S\}$. Since $S$ contains exactly one vertex from every vertex set $V_i$, $S'$ contains exactly one vertex from every vertex set $U_i$ in $G'$. We argue that $S'$ is a solution for instance $I'$. Assume to the contrary that there exists a pair $(u'_{i,i'},u'_{j,j''})\in \cal T$  such that no vertex in $S'$ satisfies its constraint. In this case let $u_{i,x}$ and $u_{j,y}$ be the vertices in $S'$ from $U_i$ and $U_{j}$ respectively. By construction $u_{i,x}$ and $u_{j,y}$ are selected in $S'$ because of  $v_{i,x}$ and $v_{j,y}$ in $S$ which are not adjacent. Since every pair in $\cal T$ correspond to an edge with endpoints in different vertex sets, either $i'$ and $x$ are distinct, or $j'$ and $y$ must be distinct, and in this case by Lemma \ref{lemma:satisfy}, constraint imposed by pair $(u'_{i,i'},u'_{j,j'})$ must be satisfied, contradicting the assumption. Thus, $S'$ is a terminal monitoring set of size $k$ for the instance $I'$.
\paragraph{}

    For the other direction, let $S'$ be a be a terminal monitoring set of size $k$ for instance $I'$, by Lemma \ref{lemma:nice}, we may assume that $S'$ is a nice solution. Thus, by definition of nice solution, $S'$ contains exactly one vertex from  $U_i$ for every $i\in[k]$, and no other vertex. We now construct $S= \{v_{i,j}\mid u_{i,j}\in S'\}$. We argue that $S$ is a solution to instance $I$. Assume to the contrary that $S$ contain two vertices $v_{i,i'}$ and $v_{j,j'}$ which are adjacent, and let $u_{i,i'}$ and $u_{j,j'}$ be their corresponding vertices in $S'$, and let $i<j$. Consider the terminal pairs $T_{i,j}$, the constraints imposed by every pair in $T_{i,j}$ must be satisfied by either $u_{i,i'}$ or $u_{j,j'}$
    because vertices in $S'\setminus \{u_{i,i'},u_{j,j'}\}$ belongs to $U_q$ where $q\neq i$ and $q\neq j$, and by Observation \ref{obs:shortest_distance_between_two_hi} every vertex in $S'\setminus \{u_{i,i'},u_{j,j'}\}$ must be $2L$ distance apart from every the vertex of every terminal pair in $T_{i,j}$, hence vertices in $S'\setminus \{u_{i,i'},u_{j,j'}\}$ can satisfy constraint imposed no terminal in $T_{i,j}$. Since $S'$ is a solution, all the constraints imposed by $T_{i,j}$ are satisfied by $u_{i,i'}$ and $u_{j,j'}$, this implies that $(u'_{i,i'},u'_{j,j'})$ is not a terminal pair in $T_{i,j}$, otherwise as per Lemma \ref{lemma:notsatisfy} its constraint would not be satisfied by either of $u_{i,i'}$ or $u_{j,j'}$ as well. Thus, $(v_{i,i'},v_{j,j'})$ is not an edge in $G$, contradicting the assumption.
    This finishes the proof.
\end{proof}
\paragraph{}

The above lemma finishes the proof of Theorem \ref{thm:hardness-fvs}.

\section{Conclusion}

We introduced {\tms} and initiated an investigation of its parameterized complexity.
We obtained both hardness and tractability results. We also studied {\atms} which is a generalization of {\tms} and proved it to be W[1]-hard parameterized by feedback vertex number plus solution size when input graph is a planar graph. We leave open the parameterized complexity of {\tms} with the following parameters: pathwidth, treewidth, feedback vertex number, vertex deletion number to disjoint union of paths.

\bibliographystyle{splncs04}
\bibliography{TMS-references}

\end{document}